\documentclass[11pt,a4paper]{article}

\usepackage{amsmath,amsfonts,amssymb,amsthm}
\usepackage{mathbbol}
\usepackage{amsthm}
\usepackage{graphicx,color}
\usepackage{boxedminipage}
\usepackage[linesnumbered, boxed, algosection,noline]{algorithm2e}
\usepackage{framed}
\usepackage{thmtools}
\usepackage{thm-restate}
\usepackage{xspace}
\usepackage{vmargin}
\setmarginsrb{1in}{1in}{1in}{1in}{0pt}{0pt}{0pt}{6mm}
\usepackage{todonotes}
 \usepackage[pdftex, plainpages = false, pdfpagelabels, 
                 bookmarks=false,
                 bookmarksopen = true,
                 bookmarksnumbered = true,
                 breaklinks = true,
                 linktocpage,
                 pagebackref,
                 colorlinks = true,  
                 linkcolor = blue,
                 urlcolor  = blue,
                 citecolor = red,
                 anchorcolor = green,
                 hyperindex = true,
                 hyperfigures
                 ]{hyperref} 
 \usepackage{xifthen}
 \usepackage{tabularx}
 \usepackage{enumitem}

 \usetikzlibrary{calc} 
 \usetikzlibrary{decorations}
 \usetikzlibrary{decorations.markings}
 \usetikzlibrary{arrows.meta}

 \newcommand{\dist}{{\operatorname{dist}}}
 \newcommand{\length}{{\operatorname{length}}}
  \newcommand{\diam}{{\operatorname{diam}}}

\DeclareMathOperator{\operatorClassP}{{\sf P}\xspace}
\newcommand{\classP}{\ensuremath{\operatorClassP}\xspace}
\DeclareMathOperator{\operatorClassNP}{{\sf NP}}
\newcommand{\classNP}{\ensuremath{\operatorClassNP}}

\DeclareMathOperator{\operatorClassFPT}{{\sf FPT}\xspace}
\newcommand{\classFPT}{\ensuremath{\operatorClassFPT}\xspace}

 \newcommand{\probKPath}{\textsc{Longest Path}\xspace}

\usepackage{empheq}
\usepackage{tcolorbox}
\definecolor{blueish}{rgb}{0.122, 0.435, 0.698}
\definecolor{dagstuhlyellow}{rgb}{0.99,0.78,0.07}
\definecolor{lightgray}{rgb}{0.9,0.9,0.9}
\newtcbox{\colbox}{
size=title,
  nobeforeafter,
  colframe=white,
  colback=lightgray,
  arc=10pt,
  tcbox raise base}
\newcommand{\defparproblemR}[4]{
  \vspace{1mm}
\noindent\colbox{
  \begin{minipage}{0.94\textwidth}
  \begin{tabular*}{\textwidth}{@{\extracolsep{\fill}}lr} #1  &
    {\textbf{Parameter:}} #3 \\ \end{tabular*}
  {\textbf{Input:}} #2  \\
  {\textbf{Task:}} #4
  \end{minipage}
  }
  \vspace{1mm}
}
\newcommand{\Oh}{\mathcal{O}}

\newtheorem{theorem}{Theorem}
\newtheorem{lemma}{Lemma}
\newtheorem{claim}{Claim}
\newtheorem{corollary}{Corollary}

\newtheorem{observation}{Observation}
\newtheorem{proposition}{Proposition}

\newcommand{\yes}{{yes}}
\newcommand{\no}{{no}}
\newcommand{\yesinstance}{\yes-instance\xspace}

\newcommand{\pname}{\textsc}
\newcommand{\ProblemFormat}[1]{\pname{#1}}
\newcommand{\ProblemIndex}[1]{\index{problem!\ProblemFormat{#1}}}
\newcommand{\ProblemName}[1]{\ProblemFormat{#1}\ProblemIndex{#1}{}\xspace}

 \newcommand{\probLD}{\ProblemName{Longest Detour}}
  \newcommand{\probED}{\ProblemName{Exact Detour}}
  \newcommand{\probLDPu}{\ProblemName{Long   $(s,t)$-Path}}
   \newcommand{\probLDP}{\ProblemName{Long Directed $(s,t)$-Path}}
   \newcommand{\probDP}{$p$-\ProblemName{Disjoint Paths}}
    \newcommand{\probDPa}{$3$-\ProblemName{Disjoint Paths}}
    \newcommand{\probLPDiam}{\ProblemName{Longest Path above Diameter}}



\newlength{\RoundedBoxWidth}
\newsavebox{\GrayRoundedBox}
\newenvironment{GrayBox}[1]%
   {\setlength{\RoundedBoxWidth}{.93\textwidth}
    \def\boxheading{#1}
    \begin{lrbox}{\GrayRoundedBox}
       \begin{minipage}{\RoundedBoxWidth}}%
   {   \end{minipage}
    \end{lrbox}
    \begin{center}
    \begin{tikzpicture}%
       \node(Text)[draw=black!20,fill=white,rounded corners,%
             inner sep=2ex,text width=\RoundedBoxWidth]%
             {\usebox{\GrayRoundedBox}};
        \coordinate(x) at (current bounding box.north west);
        \node [draw=white,rectangle,inner sep=3pt,anchor=north west,fill=white] 
        at ($(x)+(6pt,.75em)$) {\boxheading};
    \end{tikzpicture}
    \end{center}}     

\newenvironment{defproblemx}[2][]{\noindent\ignorespaces%
                                \FrameSep=6pt%
                                \parindent=0pt%
                \vspace*{-1.5em}
                \ifthenelse{\isempty{#1}}{%
                  \begin{GrayBox}{\textsc{#2}}%                
                }{%
                  \begin{GrayBox}{\textsc{#2}  parameterized by~{#1}}%  
                }
                \begin{tabular*}{\textwidth}{@{\hspace{.1em}} >{\itshape} p{1.8cm} p{0.8\textwidth} @{}}%        
            }{
                \end{tabular*}%
                \end{GrayBox}%
                \ignorespacesafterend
            }  

\begin{document}
\title{Detours in Directed Graphs\thanks{The research received funding from European Research Council (ERC) under the European Union’s Horizon 2020 research and innovation programme (grant no.~819416),  the Swarnajayanti Fellowship grant DST/SJF/MSA-01/2017-18, the Research Council of Norway via the project BWCA (grant no. 314528), and the Austrian Science Fund (FWF) via project Y1329 (Parameterized Analysis in Artificial Intelligence).
}}

\author{Fedor V. Fomin\thanks{
Department of Informatics, University of Bergen, Norway.} \addtocounter{footnote}{-1}
\and 
Petr A. Golovach\footnotemark{} \addtocounter{footnote}{-1}
\and 
William Lochet\footnotemark{} 
\and 
Danil Sagunov\thanks{St.\ Petersburg Department of V.A.\ Steklov Institute of Mathematics, Russia.}~\thanks{JetBrains Research, Saint Petersburg, Russia.} 
\and 
Kirill Simonov\thanks{Algorithms and Complexity Group, TU Wien, Austria} 
\addtocounter{footnote}{-4}
\and  
Saket Saurabh \footnotemark{}\addtocounter{footnote}{3}
\thanks{Institute of Mathematical Sciences, HBNI, Chennai, India.}
}

\date{}

\maketitle

\begin{abstract}
 We study two ``above guarantee'' versions of the classical \probKPath problem on undirected and directed graphs 
 and obtain the following results. 
 In the first variant of \probKPath that we study, called \probLD, the task is to decide whether a graph has an $(s,t)$-path of length at least $\dist_G(s,t)+k$ (where $\dist_G(s,t)$ denotes the length of a shortest path from $s$ to $t$). Bez{\'{a}}kov{\'{a}} et al.~\cite{BezakovaCDF19} proved that on undirected graphs the problem is fixed-parameter tractable (\classFPT) by providing an algorithm of running time $2^{\Oh (k)}\cdot n$. Further, they left the parameterized complexity of the problem on directed graphs open. Our first main result establishes a connection between \probLD on directed graphs and  \probDPa on directed graphs. %, where $p=3$. 
 Using these new insights, we design  a $2^{\Oh (k)}\cdot n^{\Oh(1)}$ time algorithm for the problem on directed planar graphs. Further, the new approach yields a significantly faster \classFPT algorithm on undirected graphs. 
 
 In the second variant of \probKPath, namely \probLPDiam, the task is to decide whether the graph has a path of length at least
 $\diam(G)+k$ ($\diam(G)$ denotes the length of a longest shortest path in a graph $G$). We obtain dichotomy results about \probLPDiam on undirected and directed graphs. 
For (un)directed graphs,  \probLPDiam is \classNP-complete even for $k=1$. However, if the input undirected graph is $2$-connected, then the problem is \classFPT. 
On the other hand, for $2$-connected directed graphs, we show that 
\probLPDiam is solvable in polynomial time for each $k\in\{1,\dots, 4\}$ and is \classNP-complete for every $k\geq 5$. 
The parameterized complexity of \probLD on general directed graphs remains an interesting open problem.
\end{abstract}

\section{Introduction}\label{sec:intro}
In the \probKPath problem, we are given an $n$-vertex graph $G$ and an integer~$k$. (Graph $G$ could be undirected or directed.) The task is to decide whether~$G$ contains a path of length at least $k$. \probKPath is a fundamental algorithmic problem that played one of the central roles in developing parameterized complexity~\cite{Monien85,Bodlaender93a,AlonYZ,HuffnerWZ08,KneisMRR06,ChenLSZ07,ChenKLMR09,Koutis08,Williams09,FominLS14,FominLS14,FominK13,KoutisW16,Bjorklund2017119}. To further our algorithmic knowledge about  the \probKPath problem,   Bez{\'{a}}kov{\'{a}} et al.~\cite{BezakovaCDF19} introduced  a novel ``above guarantee'' parameterization for the problem. 
For a pair of vertices $s,t$ of an $n$-vertex graph $G$, let $\dist_G(s,t)$ be the distance from $s$ to $t$, that is, the length of a shortest path from $s$ to $t$. In this variant of \probKPath, the task is to decide whether a graph has an $(s,t)$-path of length at least
$\dist_G(s,t)+k$. The difference with the ``classical'' parameterization is that instead of parameterizing by the path length, the parameterization is by the offset $k$. 

\defparproblemR{\probLD}{A graph $G$, vertices $s,t\in V(G)$, and an integer
  $k$.}{$k$}{Decide whether there is an $(s,t)$-path in $G$ of length at least
  $\dist_G(s,t) +k$.}
 
Since the length of a shortest
path between $s$ and $t$  can be found in linear time, such a parameterization could provide significantly better solutions than parameterization by the path length.
Bez{\'{a}}kov{\'{a}} et al.~\cite{BezakovaCDF19}  proved that on undirected graphs the problem is fixed-parameter tractable (\classFPT) by providing an algorithm of running time $2^{\Oh (k)}\cdot n$. Parameterized complexity of  \probLD  on directed graphs was left as the main open problem in~\cite{BezakovaCDF19}. Our paper makes significant step towards finding a solution to this open problem.

\medskip \noindent\textbf{Our results.}
Our first main result establishes a connection between  \probLD and another fundamental algorithmic problem  \probDP. 
 Recall that the   \probDP problem is to decide whether $p$ pairs of \emph{terminal} vertices $(s_i,t_i)$,  $i\in\{1,\ldots,p\}$,  in a (directed) graph $G$ could be connected by pairwise internally vertex disjoint $(s_i,t_i)$-paths.
 We prove (the formal statement of our result is given in Theorem~\ref{thm:main}) that if $\mathcal{C}$ is a class of (directed) graphs such that \probDP admits a polynomial time algorithm on $\mathcal{C}$ for $p=3$, then \probLD  is  \classFPT on $\mathcal{C}$. Moreover,    the \classFPT algorithm for \probLD    on $\mathcal{C}$ is single-exponential in $k$ (running in time $2^{\Oh (k)}\cdot n^{\Oh(1)}$). 

Unfortunately, our result does not resolve the question about parameterized complexity of  \probLD on directed graphs. Indeed,   Fortune, Hopcroft, and Wyllie~\cite{FortuneHW80} proved that \probDP is \classNP-complete on directed graphs for every fixed $p\geq 2$.  However, the new insight   helps   to establish the tractability of \probLD  on planar directed graphs, whose complexity was also  open. 
The theorem of  Schrijver from~\cite{Schrijver94} states that  \probDP
could  be solved in time $n^{\Oh(p)}$ when the input is restricted to planar directed graphs.  (This result was improved by Cygan et al.~\cite{CyganMPP13} who proved that \probDP parameterized by $p$ is \classFPT on planar directed graphs.)  Pipelined with our theorem, it immediately implies that \probLD is  \classFPT  on planar directed graphs.

Besides establishing parameterized complexity of \probLD on planar directed graphs our theorem has several  advantages over the previous work even  on undirected graphs. 
  By the seminal result of Robertson and Seymour~\cite{RobertsonS95b},   \probDP is solvable in $f(p)\cdot n^3$ time on undirected graphs for some function $f$ of $p$ only.  Therefore on undirected graphs
 \probDP  is solvable in polynomial time for every fixed $p$, and    for $p=3$ in particular. Later the result of Robertson and Seymour was improved by Kawarabayashi, Kobayashi, and Reed~\cite{KawarabayashiKR12} who gave an algorithm with quadratic dependence on the input size. Pipelined with our result, this brings us to  a Monte Carlo randomized algorithm solving \probLD on undirected graphs in time $10.8^k\cdot n^{\Oh(1)}$.  Our algorithm can be derandomized, and the deterministic algorithm runs in time $45.5^k\cdot n^{\Oh(1)}$.
 While the algorithm of Bez{\'{a}}kov{\'{a}} et al.~\cite{BezakovaCDF19} for undirected graphs runs in time $\Oh(c^k\cdot n)$, that is,   is single-exponential in $k$, the constant  $c$ is huge. The reason is that their algorithm exploits the Win/Win approach based on excluding graph minors.  More precisely, Bez{\'{a}}kov{\'{a}} et al.  proved that 
if a 2-connected graph $G$ contains as a minor,  a graph  obtained from the complete graph $K_4$ by replacing each edge by a path with $k$ edges, then $G$ has an  $(s,t)$-path of length at least $\dist_G(s,t)+k$. Otherwise, in the absence of such a graph as a minor,  the treewidth of $G$ is at most $32k+46$.  
Combining this fact with an FPT 3-approximation  algorithm~\cite{BodlaenderDDFLP16}, running in time $2^{\Oh(k)}\cdot n^{\Oh(1)}$,  to compute the treewidth of a graph,  brings us to a graph of treewidth at most $96k+\Oh(1)$. Finally,   solving \probLD on graphs of bounded treewidth by one of the known single-exponential algorithms, see ~\cite{cut-and-count,BodlaenderCKN15,FominLPS17}, will result in running time 
$3^{96k}\cdot n^{\Oh(1)}$. Thus on undirected graphs,  our algorithm reduces the constant $c$ in the base of the exponent from $3^{96}$ down to $10.8$!

\medskip

Our second set of results addresses the complexity of the problem strongly related to  \probLD. 
 The length of a longest shortest path in a graph $G$ is denoted by  \emph{diameter of $G$},  $\diam(G)$. Thus every graph $G$ has a path of length at least  $\diam(G)$. But does it have a path of length longer than  $\diam(G)$? This leads to 
the following parameterized problem. 

\defparproblemR{\probLPDiam}{A graph $G$ and an integer
  $k$.}{$k$}{Decide whether there is a path in $G$ of length at least
 $\diam(G)+k$.}

 As in \probLD, the parameterization is by the offset $k$. 
When $(s,t)$ is a pair of   diametral vertices in $G$, the length of the shortest $(s,t)$-path in $G$ is the diameter of $G$.
However, this does not allow to reduce \probLPDiam to \probLD --- if there is a path of length  $\diam(G)+k$ in $G$, it is not necessarily an $(s,t)$-path. Moreover, such a path might connect two vertices with a much smaller distance between them than $\diam(G)$. 
In fact, our hardness results for \probLPDiam are based precisely on instances where the target path has this property: its length is very close to $\diam(G)$, but much larger than the shortest distance between its endpoints.
Thus, the lower bounds we obtain for \probLPDiam are not applicable to \probLD.

We obtain the following dichotomy results about \probLPDiam on undirected and directed graphs. 
For undirected graphs,   \probLPDiam is \classNP-complete even for $k=1$. However, if the input undirected graph is $2$-connected, that is, it remains connected after deleting any of its vertices, then the problem is \classFPT. 
For directed graphs, the problem  is also \classNP-complete even for $k=1$.  However, the situation is more complicated  and interesting on $2$-connected directed graphs. (Let us remind that a strongly connected digraph $G$ is $2$-connected or strongly $2$-connected, if for every vertex $v\in V(G)$, graph $G-v$ remains strongly connected.) In this case, we show that 
\probLPDiam is solvable in polynomial time for each $k\in\{1,\dots, 4\}$ and is \classNP-complete for every $k\geq 5$.

\medskip \noindent\textbf{Our approach.} 
 A natural way to approach  \probLD on directed graphs would be to mimic the algorithm for undirected graphs.  By the result of Kawarabayashi and Kreutzer~\cite{KawarabayashiK15}, every directed graph of sufficiently large directed treewidth contains a sizable directed grid as a ``butterfly minor''. 
However, as reported in \cite{BezakovaCDF17}, there are several obstacles towards applying the grid theorem of Kawarabayashi and Kreutzer for obtaining a Win/Win algorithm. 
After several unsuccessful attempts, we switched to another strategy.

We start the proof of Theorem~\ref{thm:main} by checking whether  $G$ has an $(s,t)$-path of length $\dist_G(s,t)+\ell$ for $k\leq \ell<2k$.
This can be done in time $2^{\Oh(k)}\cdot n^{\Oh(1)}$ by calling the algorithm of Bez{\'{a}}kov{\'{a}} et al.~\cite{BezakovaCDF19} that 
 finds an $(s,t)$-path in a directed $G$ of length \emph{exactly} $\dist_G(s,t)+\ell$. 
If such a path is not found, we conclude that if $(G,k)$ is a \yesinstance, then $G$ contains  an $(s,t)$-path of length at least $\dist_G(s,t)+2k$.

Next, we check whether there exist two vertices $v$ and $w$ reachable from $s$ such that $\dist_G(s,w)-\dist_G(s,v)\geq k$ and $G$ has pairwise disjoint $(s,w)$-, $(w,v)$-, and $(v,t)$-paths. If such a pair of vertices exists, we obtain a solution by  concatenating  disjoint $(s,w)$-, $(w,v)$-, and $(v,t)$-paths. This is the place in our algorithm, where we require a subroutine solving \textsc{$3$-Disjoint Paths}.

 When none of the above procedures finds a detour, we prove a combinatorial claim that allows reducing the search of a solution to a significantly smaller region of the graph. This combinatorial claim is the essential part of our algorithm. More precisely, we show that there are two vertices $u$ and $x$, and a specific induced subgraph $H$ of $G$ (depending on $u$ and $x$) such that $G$ has an $(s,t)$-path of length at least $\dist_G(s,t)+k$ if and only if $H$ has an $(u,x)$-path of length at least $\ell$ for a specific $\ell\leq 2k$ (also depending on $u$ and $x$). Moreover, given $u$,  in polynomial time, we can find a feasible domain for vertex $x$, and for each choice of $x$, we can also determine $\ell$ and construct $H$ in polynomial time. Then we apply the algorithm of Fomin et al.~\cite{FominLPSZ18} to check whether $H$ has an $(u,x)$-path in $H$ of length at least $\ell$.

 Our strategy for \probLPDiam is different. For undirected graphs, the solution turns out to be reasonably simple. It easy to show that \probLPDiam is 
\classNP-complete for $k=1$ by reducing 
 \textsc{Hamiltonian Path} to it. When an undirected graph $G$ is $2$-connected, and the diameter is larger than $k+1$, then $G$ always contains a path of length at least $d+k$. If the diameter is at most $k$, it suffices to run a \probKPath algorithm to show that the problem is \classFPT. For directed graphs, a similar reduction shows that the problem is  \classNP-complete for $k=1$. However, for  2-strongly-connected directed graphs, the situation is much more interesting. It is not too difficult to prove that when the diameter of a 2-strongly-connected digraph is sufficiently large, it always contains a path of length $\diam(G)+1$. 
 With much more careful arguments, it is possible to push this up to $k=4$. Thus for each $k\leq 4$, the problem is solvable in polynomial time. For $k=5$ we can construct a family of 2-strongly-connected digraphs of arbitrarily large diameter that do not have a path of length $\diam(G)+5$. These graphs become extremely useful as gadgets that we use to prove that the problem is \classNP-complete for each $k\geq 5$.

 \medskip \noindent\textbf{Related work.} 
There is a vast literature in the field of parameterized complexity devoted to \probKPath 
\cite{Monien85,Bodlaender93a,AlonYZ,HuffnerWZ08,KneisMRR06,ChenLSZ07,ChenKLMR09,Koutis08,Williams09,FominLS14,FominLS14,Bjorklund2017119}. The surveys \cite{FominK13,KoutisW16} and the textbook
\cite[Chapter~10]{cygan2015parameterized} provide an  overview of the advances in the area.

\probLD was introduced by Bez{\'{a}}kov{\'{a}} et al.  in~\cite{BezakovaCDF19}. They gave an \classFPT algorithm for undirected graphs and posed the question about detours in directed graphs. Even the existence of a polynomial time algorithm for   \probLD with $k=1$, that is, deciding whether a directed graph has a path longer than a shortest $(s,t)$-path, is open. For the related \probED problem, deciding whether there is a detour of length \emph{exactly}  $\dist_G(s,t) +k$  is \classFPT both on directed and undirected graphs~\cite{BezakovaCDF19}.

Another problem related to our work is \probLDPu. Here for vertices $s$ and $t$ of  a graph $G$, and integer parameter $k$, we have to 
decide whether there is an $(s,t)$-path in $G$ of length at least
  $k$.  A simple trick, see \cite[Exercise~5.8]{cygan2015parameterized}, allows to use color-coding to show   that \probLDPu  is \classFPT on undirected graph. For directed graphs the situation is more involved, and the first \classFPT algorithm for   \probLDPu on directed graphs was obtained only recently~\cite{FominLPSZ18}.
  The proof of Theorem~\ref{thm:main} uses some of the ideas developed in~\cite{FominLPSZ18}.

  Both \probLD and \probLPDiam fit into the research subarea of 
  parameterized complexity called 
 ``above guarantee'' parameterization~\cite{MahajanR99,AlonGKSY10,CrowstonJMPRS13,GargP16,DBLP:journals/mst/GutinKLM11,GutinIMY12,GutinP16,GutinRSY07,LokshtanovNRRS14,MahajanRS09}.
  Besides the work of  Bez{\'{a}}kov{\'{a}} et al.~\cite{BezakovaCDF17},  several papers study parameterization of longest paths and cycles above different guarantees.
 Fomin et al.  \cite{fomin_et_al:LIPIcs:2020:11724}  designed parameterized algorithms for computing paths and cycles longer than the girth of a graph. The same set of the authors in~\cite{FominGLPSZ20}  studied \classFPT algorithms that finds  paths and cycles above degeneracy. Fomin et al.  \cite{fominGSS20Dirac} developed an \classFPT algorithm computing a cycle of length $2\delta +k$, where $\delta$ is the minimum vertex degree of the input graph.  
 Jansen, Kozma, and Nederlof  in \cite{DBLP:conf/wg/Jansen0N19}  looked at parameterized complexity of Hamiltonicity below Dirac's conditions.  Berger,   Seymour, and Spirkl in~\cite{berger2020finding},  gave a polynomial time algorithm that, with input a graph $G$ and two vertices $s, t$ of $G$, that decides whether there is an \emph{induced} $(s,t)$-path that is longer than a shortest $(s,t)$-path. All these algorithms for computing long paths and cycles above some guarantee are for undirected graphs.

\medskip
The remaining part of this paper is organized as follows. In Section~\ref{sec:prelim}, we give preliminaries. 
In Section~\ref{sec:algorithm}, we prove our first main result establishing connections between
\probDPa and \probLD (Theorem~\ref{thm:main}). 
 Section~\ref{sec:LPAD} is devoted to \probLPDiam. The concluding Section~\ref{sec:concl} provides open questions for further research.

\section{Preliminaries}\label{sec:prelim} 
 \noindent\textbf{Parameterized Complexity.}
We refer to the recent books~\cite{CyganFKLMPPS15,DowneyF13} for the detailed introduction to Parameterized Complexity. Here we just remind that the  computational complexity of an algorithm solving a parameterized problem is measured as a function of the input size $n$ of a problem and an integer \emph{parameter} $k$ associated with the input.
A parameterized problem is said to be \emph{fixed-parameter tractable} (or \classFPT) if it can be solved in time $f(k)\cdot n^{\Oh(1)}$ for some function~$f(\cdot)$.

\medskip
 \noindent\textbf{Graphs.}
Recall that an undirected graph is a pair $G=(V,E)$, where $V$  is a set of vertices and $E$ is a set of unordered pairs $\{u,v\}$ of distinct vertices called \emph{edges}.  A directed graph $G=(V,A)$ is a pair, where $V$ is a set of vertices and 
$A$ is a set of ordered pairs $(u,v)$ of distinct vertices called \emph{arcs}. Note that we do not allow loops and multiple edges or arcs.
We use $V(G)$ and $E(G)$ ($A(G)$, respectively) to denote the set of vertices and the set of edges (set of arcs, respectively) of $G$. 
We write $n$ and $m$ to denote the number of vertices and edges (arcs, respectively) if this does not create confusion. 
For a (directed) graph $G$ and a subset $X\subseteq V(G)$ of vertices, we write $G[X]$ to denote the subgraph of $G$ induced by $X$.
For a set of vertices $S$, $G-S$ denotes the (directed) graph obtained by deleting the vertices of $S$, that is, $G-S=G[V(G)\setminus S]$.
We write $P=v_1\cdots v_k$ to denote a \emph{path} with the vertices $v_1,\ldots,v_k$ and the edges  $\{v_1,v_2\},\ldots,\{v_{k-1},v_k\}$ (arcs $(v_1,v_2),\ldots,(v_{k-1},v_k)$, respectively);
$v_1$ and $v_k$ are the \emph{end-vertices} of $P$ and the vertices $v_2,\dots,v_{k-1}$ are \emph{internal}. We say that $P$ is an \emph{$(v_1,v_k)$-path}. The \emph{length} of $P$, denoted by $\length(P)$, is the number of edges (arcs, respectively) in $P$.
Two paths are \emph{disjoint} if they have no common vertex and they are \emph{internally disjoint} if no internal vertex of one path is a vertex of the other.
For a $(u,v)$-path $P_1$ and a $(v,w)$-path $P_2$ that are internally disjoint, we denote by $P_1\circ P_2$ the \emph{concatenation} of $P_1$ and $P_2$. 
A vertex $v$ is \emph{reachable} from a vertex $u$ in a (directed) graph $G$ if $G$ has a $(u,v)$-path. 
For $u,v\in V(G)$, $\dist_G(u,v)$ denotes the \emph{distance} between $u$ and $v$ in $G$, that is, the minimum number of edges (arcs, respectively) in an $(u,v)$-path. An undirected graph $G$ is \emph{connected} if for every two vertices $u$ and $v$, $G$ has a $(u,v)$-path. A directed graph $G$ is \emph{strongly-connected} if for every two vertices $u$ and $v$ both $u$ is reachable form $v$ and $v$ is reachable from $u$. For a positive integer $k$, an undirected (directed, respectively) graph $G$ is \emph{$k$-connected} (\emph{$k$-strongly-connected}, respectively) if $|V(G)|\geq k$ and $G-S$ is connected (strongly-connected, respectively) for every $S\subseteq V(G)$ of size at most $k-1$.
For a directed graph $G$, by $G^T$ we denote the \emph{transpose} of $G$, i.e.\ $G^T$ is a directed graph defined on the same set of vertices and the same set of arcs, but the direction of each arc in $G^T$ is reversed.

\medskip

We use several known parameterized algorithms for finding long paths. First of all, let us recall the currently fastest deterministic algortihm for \probKPath on directed graphs due to Tsur~\cite{Tsur19b}.

\begin{proposition}[\cite{Tsur19b}]\label{prop:KPath}
There is a deterministic algorithm for \probKPath with running time $2.554^k\cdot n^{\Oh(1)}$.
\end{proposition}

We also need the result of Fomin et al.~\cite{FominLPSZ18} for the \probLDP problem. This problem asks, given a directed graph $G$, two vertices $s,t\in V(G)$, and an integer $k\geq 0$, whether $G$ has an $(s,t)$-path of length at least $k$. 

\begin{proposition}[\cite{FominLPSZ18}]\label{prop:LDP}
\probLDP can be deterministically solved in time $4.884^k\cdot n^{\Oh(1)}$.
\end{proposition}

Clearly, both results holds for the variant of the problem on undirected graphs.

Finally, we use the result of Bez{\'{a}}kov{\'{a}} et al.~\cite{BezakovaCDF19} for the variant of \probLD whose task is, given a (directed) graph $G$, two vertices $s,t\in V(G)$, and an integer $k\geq 0$, decide whether $G$ has an $(s,t)$-path of length \emph{exactly} $\dist_G(s,t)+k$.

\begin{proposition}[\cite{BezakovaCDF19}]\label{prop:exact-detour}
There is a bounded-error randomized algorithm that solves \probED on undirected  graphs in time $2.746^k\cdot n^{\Oh(1)}$ and on directed graphs
in time $4^k\cdot n^{\Oh(1)}$. For both undirected and directed graphs, there is a deterministic
algorithm that runs in time $6.745^k\cdot n^{\Oh(1)}$.
\end{proposition}

\section{An FPT algorithm for finding detours}\label{sec:algorithm}
In this section we show the first main result of our paper. 

\begin{theorem}\label{thm:main}
Let $\mathcal{C}$ be a class of directed graphs such that \probDPa 
can be solved in $f(n)$ time 
time on $\mathcal{C}$. Then \probLD can be solved in $45.5^k\cdot n^{\Oh(1)}+\Oh(f(n)n^2)$ time by a deterministic algorithm and in $23.86^k\cdot n^{\Oh(1)}+\Oh(f(n)n^2)$ time by a bounded-error randomized algorithm
when the input is restricted to graphs from $\mathcal{C}$.
\end{theorem}

\begin{proof}
Let $(G,s,t,k)$ be an instance of \probLD with $G\in \mathcal{C}$. For $k=0$, the problem is trivial and we assume that $k\geq 1$. We also have that  $(G,s,t,k)$ is a trivial no-instance if $t$ is not reachable from $s$.
We assume from now that every vertex of $G$ is reachable from $s$. Otherwise, we set $G:=G[R]$, where $R$ is the set of vertices of $G$ reachable from $s$ using the straightforward property that  every $(s,t)$-path in $G$ is a path in $G[R]$. Clearly, $R$ can be constructed in $\Oh(n+m)$ time by the breadth-first search. 

Using Proposition~\ref{prop:exact-detour}, we check in $6.745^{2k}\cdot n^{\Oh(1)}$ time by a deterministic algorithm (in $4^{2k}\cdot n^{\Oh(1)}$ time by a randomized algorithm, respectively) 
whether $G$ has an $(s,t)$-path of length $\dist_G(s,t)+\ell$ for some $k\leq \ell\leq 2k-1$ by trying all values of $\ell$ in this interval. We return a solution and stop if we discover such a path. Assume from now that this is not the case, that is, if $(G,s,t)$ is a yes-instance, then the length of every $(s,t)$-path of length at least $\dist_G(s,t)+k$ is at least  $\dist_G(s,t)+2k$.

We perform the breadth-first search from $s$ in $G$. For an integer $i\geq 0$, denote by $L_i$ the set of vertices at distance $i$ from $s$. 
Let $\ell$ be the maximum index such that $L_\ell\neq\emptyset$. 
Because every vertex of $G$ is reachable from $s$, $V(G)=\bigcup_{i=0}^\ell L_i$. We call $L_0,\ldots,L_\ell$ \emph{BFS-levels}.

\begin{figure}[ht]
\centering
\scalebox{0.7}{
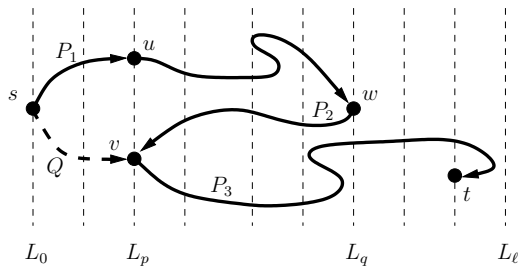}
\caption{The choice of the BFS-levels $L_p$ and $L_q$, vertices $u$, $v$, and $w$, and the  paths $P_1$, $P_2$, and $P_3$. }
\label{fig:struct}
\end{figure}

Our algorithm is based on structural properties of potential solutions. Suppose that $(G,s,t,k)$ is a yes-instance and let a path $P$ be a solution of minimum length, that is, $P$
 is an $(s,t)$-path of length at least $\dist_G(s,t)+k$ and among such paths the length of $P$ is minimum. Denote by $p\in\{1,\ldots,\ell\}$ the minimum index such that $L_p$ contains at least two vertices of $G$. Such an index exists, because if $|V(P)\cap L_i|\leq 1$ for all $i\in\{1,\ldots,\ell\}$, then $P$ is a shortest $(s,t)$-path by the definition of $L_0,\ldots,L_\ell$ and the length of $P$ is $\dist_G(s,t)<\dist_G(s,t)+k$ as $k\geq 1$.  
Let $u$ be the first (in the path order) vertex of $P$ in $L_p$ and let $v\neq u$ be the second vertex of $P$ that occurs in $L_p$. Denote by $P_1$, $P_2$, and $P_3$ the $(s,u)$, $(u,v)$, and $(v,t)$-subpath of $P$, 
respectively. Clearly, $P=P_1\circ P_2\circ P_3$. Let $q\in\{p,\ldots,\ell\}$ be the maximum index such that $P_2$ contains a vertex of $L_q$. Then denote by $w$ the first vertex of $P_2$ in $L_q$.  
See Figure~\ref{fig:struct} for the illustration of the described configuration. We use this notation for a (hypothetical) solution throughout the proof of the theorem.
The following claim is crucial for us. 

\begin{claim}\label{cl:first-detour}
The length of $P_2$ is at least $k$. 
\end{claim} 
 
 \begin{proof}[Proof of Claim~\ref{cl:first-detour}]
For the sake of contradiction, assume that the length of $P_2$ is less than $k$. Let $Q$ be a shortest $(s,v)$-path in $G$. By the definition of BFS-levels, $V(Q)\subseteq L_0\cup\cdots\cup L_p$ and $v$ is a unique vertex of $Q$ in $L_p$. This implies that $Q$ is internally vertex disjoint with $P_3$. Note that the length of $Q$ is the same as the length of $P_1$, because $P_1$ contains exactly one vertex from each of the BFS levels $L_1,\ldots,L_p$. Then $P'=Q\circ P_3$ is an $(s,t)$-path and 
\begin{align*}
\length(P')=&\length(Q)+\length(P_3)=\length(P_1)+\length(P_3)\\=&\length(P)-\length(P_2)\leq \length(P)-k.
\end{align*}
Recall that the length of every $(s,t)$-path of length at least $\dist_G(s,t)+k$ is at least  $\dist_G(s,t)+2k$. This means that $\length(P)\geq \dist_G(s,t)+2k$ and, therefore, the length of $P'$ is at least $\dist_G(s,t)+k$, that is, $P'$ is a solution to the considered instance. However, $\length(P')<\length(P)$, because $P_2$ contains at least one arc. This contradicts the choice of $P$ as a solution of minimum length. This completes the proof of the claim.
\end{proof}
 
By Claim~\ref{cl:first-detour}, solving  \probLD  on $(G,s,t,k)$ boils down to identifying internally disjoint $P_1$, $P_2$, and $P_3$, where the length of $P_2$ is at least $k$. 

First, we check whether we can find paths for $q-p\geq k-1$.  Notice that if $q-p\geq k-1$, then for every internally disjoint $(s,w)$-, $(w,v)$-, and $(v,t)$-paths $R_1$, $R_2$, and $R_3$ respectively, their concatenation $R_1\circ R_2\circ R_3$ is an $(s,t)$-path of length at least $\dist_G(s,t)+k$.  Recall that $G\in \mathcal{C}$ and \probDP can be solved in polynomial time on this graph class for $p=3$. For every choice of two vertices $w,v\in V(G)$, we solve  \probDP on the instance $(G,(s,w),(w,v),(v,s))$. Then if there are paths $R_1$, $R_2$, and $R_3$ forming a solution to this instance, we check whether $\length(R_1)+\length(R_2)+\length(R_3)\geq \dist_G(s,t)+k$. If this holds, we conclude that the path $R_1\circ R_2\circ R_3$ is a solution to the instance $(G,s,t,k)$ of \probLD and return it. Assume from now that this is not the case, that is, we failed to find a solution of this type. Then we can complement Claim~\ref{cl:first-detour} by the following observation about our hypothetical solution $P$.

\begin{claim}\label{cl:first-detour-bounds}
$q-p\leq k-2$. 
\end{claim} 

This means that we can assume that $k\geq 2$ and 
have to check whether we can identify $P_1$, $P_2$, and $P_3$, where $V(P_2)\subseteq \bigcup_{i=p}^{p+k-2}L_i$. For this, we go over all possible choices of $u$. Note that the choice of $u$ determines $p$, i.e., the index of the BFS-level containing $u$. We consider the following two cases for each considered choice of $u$.

\begin{figure}[ht]
\centering
\scalebox{0.7}{
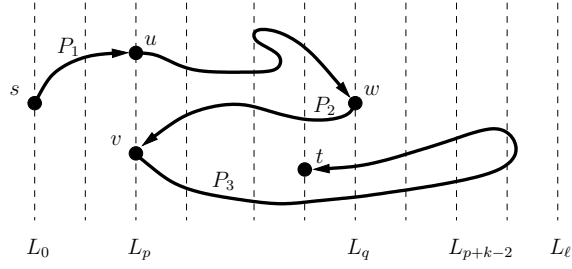}
\caption{The structure of paths  $P_1$, $P_2$, and $P_3$ in Case~1. }
\label{fig:struct-one}
\end{figure}

 \noindent\textbf{Case~1.}
$t\in L_r$ for some $p\leq r\leq p+k-2$ (see Figure~\ref{fig:struct-one}).  Then $\dist_G(s,t)=r$ and $(G,s,t,k)$ is a yes-instance if and only if $G[L_p\cup\cdots\cup L_\ell]$ has a $(u,t)$-path $S$ of length at least $(r-p)+k$, because the $(s,u)$-subpath of a potential solution should be a shortest $(s,u)$-path. Since $r-p\leq k-2$, we have that $(r-p)+k\leq 2k-2$ and we can find $S$ in  $4.884^{2k}\cdot n^{\Oh(1)}$ time by Proposition~\ref{prop:LDP} if it exists. If we obtain $S$, then we consider an arbitrary shortest $(s,u)$-path $S'$ in $G$ and conclude that $S'\circ S$ is a solution. This completes Case~1.

\begin{figure}[ht]
\centering
\scalebox{0.7}{
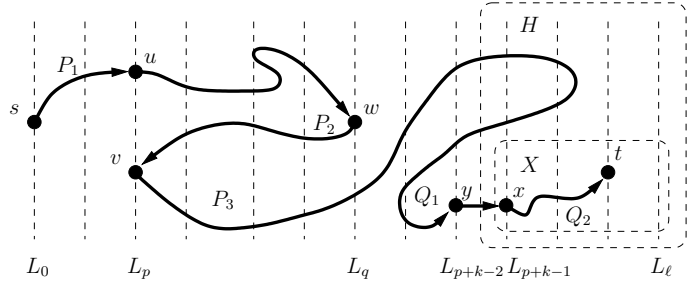}
\caption{The structure of paths  $P_1$, $P_2$, and $P_3$ in Case~2. }
\label{fig:struct-two}
\end{figure}

 \noindent\textbf{Case~2.}
$t\in L_r$ for some $r\geq p+k-1$ (see Figure~\ref{fig:struct-two}). We again consider our hypothetical solution $P=P_1\circ P_2\circ P_3$. Let $H=G[L_{p+k-1}\cup\cdots\cup L_\ell]$. Denote by $X$ the set of vertices $x\in V(H)$  such that $t$ is reachable from $x$ in $H$. 
Denote by $x$ the first vertex of $P_3$ in $X$. Clearly, such a vertex exists because $t\in X$. Moreover, $x\in L_{p+k-1}$ and its predecessor $y$ in $P_3$ is in $L_{p+k-2}$. Otherwise, $t$ would be reachable from $y\in V(H)$ in $H$ contradicting the choice of $x$.  Let $Q_1$ and $Q_2$ be the $(v,y)$- and $(x,t)$-subpaths of $P_3$. Then $P_3=Q_1\circ yx\circ Q_2$. We show one more claim about the hypothetical solution $P$.

\begin{claim}\label{cl:first-detour-reach}
$V(Q_1)\cap X=\emptyset$. 
\end{claim} 

\begin{proof}[Proof of Claim~\ref{cl:first-detour-reach}]
The proof is by contradiction. Assume that $z\in V(Q_1)\cap X$. Then $t$ is reachable from $z$ in $H$. However, $x$ is the first vertex of $P_3$ with this property by the definition; a contradiction.
\end{proof}

Notice that because $x\in X$, there is an $(x,t)$-path $Q_2'$ with $V(Q_2')\subseteq X$. By Claim~\ref{cl:first-detour-reach}, $Q_1$ and $Q_2'$ are disjoint. Since $X\subseteq L_{p+k-1}\cup\cdots\cup L_\ell$, we have that $(V(P_1)\cup V(P_2))\cap X=\emptyset$. In particular,
$Q_2'$ is disjoint with $P_1$ and $P_2$ as well. Let $P_3'=Q_1\circ yx\circ Q_2'$. By Claim~\ref{cl:first-detour}, $P'=P_1\circ P_2\circ P_3'$ is a solution, because $\length(P_2)\geq k$.  This allows us to conclude that $(G,s,t,k)$ has a solution (for the considered choice of $u$) if and only if there is $y\in L_{p+k-2}$ such that 
\begin{itemize}
\item[(i)] there is $x\in X$ such that $(y,x)\in A(G)$, and
\item[(ii)] the graph $G[L_p\cup\cdots\cup L_{\ell}]-X$ has a $(u,y)$-path of length at least $2k-2$.
\end{itemize}

Our algorithm proceeds as follows. We construct the set $X$ using the breadth-first search in $\Oh(n+m)$ time. Then for every $y\in  L_{p+k-2}$ we check (i) whether  there is $x\in X$ such that $(y,x)\in A(G)$, and (ii) whether  $G[L_p\cup\cdots\cup L_{\ell}]-X$ has a $(u,y)$-path $S$ of length at least $2k-2$. To verify (ii), we apply Proposition~\ref{prop:LDP} that allows to perform the check in $4.884^{2k}\cdot n^{\Oh(1)}$ time.
If we find such a vertex $y$ and path $S$, then to obtain a solution, we consider an arbitrary shortest $(s,u)$-path $S'$ and an arbitrary $(x,t)$ path $S''$ in $G[X]$. Then $P'=S'\circ S\circ yx\circ S''$ is a required solution to $(G,s,t,k)$. This concludes the analysis in Case~2 and the construction of the algorithm. 

\medskip
The correctness of our algorithm has been proved simultaneously with its construction. The remaining task is to evaluate the total running time. Recall that we verify  in $6.745^{2k}\cdot n^{\Oh(1)}$ time whether $G$ has an $(s,t)$-path of length $\dist_G(s,t)+\ell$ for some $k\leq \ell\leq 2k-1$ by a deterministic algorithm, and we need $4^{2k}\cdot n^{\Oh(1)}$ time if we use a randomized algorithm. Then we construct the BFS-levels in linear time. Next, we consider $\Oh(n^2)$ choices of $v$ and $w$ and apply the algorithm for 
 \probDPa $(G,(s,w),(w,v),(v,s))$ in $f(n)$ time. If we failed to find a solution so far, we proceed with $\Oh(n)$ possible choices of $u$ and consider either Case~1 or 2 for each choice. In Case~1, we solve the problem in $4.884^{2k}\cdot n^{\Oh(1)}$ time. In Case~2, we construct $X$ in $\Oh(n+m)$ time. Then for $\Oh(n)$ choices of $y$, we verify conditions (i) and (ii) in  $4.884^{2k}\cdot n^{\Oh(1)}$ time. Summarizing, we obtain that the total running time is  $6.745^{2k}\cdot n^{\Oh(1)}+\Oh(f(n)n^2)$.  Because $6.745^2<45.5$, we have that the deterministic algorithm runs in $45.5^k\cdot n^{\Oh(1)}+\Oh(f(n)n^2)$ time. Since $4^2<4.884^2<23.86$, we conclude that the problem can be solved in $23.86^k\cdot n^{\Oh(1)}+\Oh(f(n)n^2)$ time by a bounded-error randomized algorithm. 
 \end{proof}

In particular, combining Theorem~\ref{thm:main} with the results of Cygan et al.~\cite{CyganMPP13}, we obtain the following corollary.

\begin{corollary}\label{cor:planar}
 \probLD can be solved  in $45.5^k\cdot n^{\Oh(1)}$ time by a deterministic algorithm and in $23.86^k\cdot n^{\Oh(1)}$ time by a bounded-error randomized algorithm on planar directed graphs.
 \end{corollary} 

Using the fact that \probDP can be solved in $\Oh(n^2)$ time by the results of Kawarabayashi, Kobayashi, and Reed~\cite{KawarabayashiKR12}, we immediately obtain the result for \probLD on undirected graphs. However, we can improve the running time of a randomized algorithm by tuning our algorithm for the undirected case.

\begin{corollary}\label{cor:undir}
 \probLD can be solved in $45.5^k\cdot n^{\Oh(1)}$ time by a deterministic algorithm and in $10.8^k\cdot n^{\Oh(1)}$ time by a bounded-error randomized algorithm on undirected graphs.
 \end{corollary}
 
\begin{proof} 
The deterministic algorithm is the same as in the directed case. To obtain a better randomized algorithm, 
we follow the algorithm from Theorem~\ref{thm:main} and use the notation introduced in its proof.  Let $(G,s,t,k)$ be an instance of \probLD with $G\in \mathcal{C}$. We assume without loss of generality that $k\geq 1$ and $G$ is connected. 
Using Proposition~\ref{prop:exact-detour}, we check in $2.746^{2k}\cdot n^{\Oh(1)}$ time by a randomized algorithm 
whether $G$ has an $(s,t)$-path of length $\dist_G(s,t)+\ell$ for some $k\leq \ell\leq 2k-1$. If we fail to find a solution this way, we construct the BFS-levels $L_0,\ldots,L_\ell$. 

Suppose that $(G,s,t,k)$ is a yes-instance with a hypothetical solution $P$ composed by the concatenation of $P_1$, $P_2$, and $P_3$ as in the proof of Theorem~\ref{thm:main}. Let also $L_p$ and $L_q$ be the corresponding BFS-levels. 
Observe that if $q-p\geq k/2$, then $\length(P_2)\geq k$, because for every edge $\{x,y\}$ of $G$,  $x$ and $y$ are either in the same BFS-level or in consecutive levels contrary to the directed case where we may have an arc $(x,y)$ where $x\in L_i$ and $y\in L_j$ for arbitrary $j\in\{0,\ldots,i\}$. Recall that for every choice of two vertices $w,v\in V(G)$, we solve  \probDP on the instance $(G,(s,w),(w,v),(v,s))$ and try to find a solution 
to $(G,s,t,k)$ by concatenating the solutions for these instances of  \probDP. If we fail to find a solution this way, we can conclude now that $q-p\leq k/2-1$ improving Claim~\ref{cl:first-detour-bounds}.
Further, we pick $u$ and consider two cases.

In Case~1, where $t\in L_r$ for some $p\leq r\leq p+k/2-1$, we now find a $(u,t)$-path $S$ in $G[L_p\cup\cdots\cup L_\ell]$  of length at least $(r-p)+k\leq 3k/2$  in $4.884^{3k/2}\cdot n^{\Oh(1)}$ time. If such a path exists, we obtain a solution.

In Case~2, where $t\in L_r$ for some $r\geq p+k/2$,  we consider $H=G[L_{h+1}\cup\cdots\cup L_\ell]$ for $h=p+\lceil k/2\rceil$  and denote by $X$ the set of vertices of the connected component of $H$ containing $X$.  
Then for every $y\in  L_{h}$ we check (i) whether  there is $x\in X$ such that $\{y,x\}\in E(G)$, and (ii) whether  $G[L_p\cup\cdots\cup L_{\ell}]-X$ has a $(u,y)$-path $S$ of length at least $k+ \lceil k/2\rceil$ in 
$4.884^{3k/2}\cdot n^{\Oh(1)}$ time. If such a path exists, we construct a solution containing it in the same way as on the directed case.

The running time analysis is essentially the same as in the proof of Theorem~\ref{thm:main}. The difference is that now we have that $2.746^2\leq 4.884^{3/2}<10.80$. This implies that the algorithm runs in $10.8^k\cdot n^{\Oh(1)}$ time.
\end{proof}

\section{\textsc{Longest Path Above Diameter}}\label{sec:LPAD}

In this section, we investigate the complexity of \probLPDiam.
It can be noted that this problem is \classNP-complete in general even for $k=1$.

\begin{proposition}\label{prop:diam-hard}
	\probLPDiam is \classNP-complete for $k=1$ on undirected graphs. 
\end{proposition}

\begin{proof}
	Let $G$ be an undirected graph with $n\geq 2$ vertices. We construct the graph $G'$ as follows (see Figure~\ref{fig:hard}).
	\begin{itemize}
		\item Construct a copy of $G$.
		\item Construct a vertex $u$ and make it adjacent to every vertex of the copy of $G$.
		\item Construct two vertices $s$ and $t$, and then $(s,u)$ and $(u,t)$ paths $P_s$ and $P_t$, respectively, of length $n-1$.   
	\end{itemize}
	
	\begin{figure}[ht]
		\centering
		\scalebox{0.7}{
			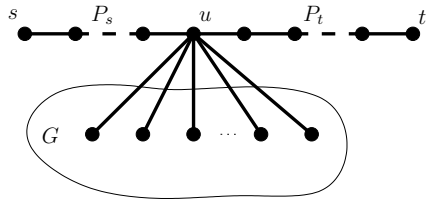}
		\caption{Construction of $G'$. }
		\label{fig:hard}
	\end{figure} 
	
	Notice that $\diam(G)=\length(P_s)+\length(P_t)=2n-2$. It is easy to verify that $G'$ has a path of length $2n-1$ if and only if $G$ has a path of length $n-1$, that is, $G$ is Hamiltonian. Because \textsc{Hamiltonian Path} is well-known to be \classNP-complete~\cite{GareyJ79}, we conclude that  \probLPDiam is \classNP-complete for $k=1$
\end{proof}

Proposition~\ref{prop:diam-hard} immediately implies that \probLPDiam is \classNP-complete for $k=1$ on strongly connected directed graphs as we can reduce the problem on undirected graphs to the directed variant by replacing each edge by the pair of arcs with opposite orientations. Still, it can be observed that the reduction in  Proposition~\ref{prop:diam-hard}  strongly relies on the fact that the constructed graph $G'$ has an articulation point $u$. Hence, it is natural to investigate the problem further imposing connectivity constraints on the input graphs. And indeed, it can be easily seen that  \probLPDiam is \classFPT on 2-connected undirected graphs.

\begin{observation}\label{obs:diam-FPT}
	\probLPDiam can be solved in time  $6.523^k\cdot n^{\Oh(1)}$ on undirected 2-connected graphs.
\end{observation}   

\begin{proof}
	Let $(G,k)$ be an instance of   \probLPDiam where $G$ is 2-connected. If $d=\diam(G)\leq k$, we can solve the problem in time $2.554^{d+k}\cdot n^{\Oh(1)}$ by using the algorithm of Proposition~\ref{prop:KPath} to check whether $G$ has a path of length $d+k$. Note that $2.554^{d+k}\leq 2.554^{2k}\leq 6.523^k$. Otherwise, if $d>k$, consider a pair of vertices $s$ and $t$ with $\dist_G(s,t)=d$. Because $G$ is 2-connected, by Menger's theorem (see, e.g.,~\cite{Diestel12}),  $G$ has a cycle $C$ containing $s$ and $t$. Since  $\dist_G(s,t)=d$ and $d\geq k+1$, the length of $C$
	is at least $d+k+1$. This implies that $C$ contains a path of length $d+k$.
\end{proof} 

However, the arguments from the proof of Observation~\ref{obs:diam-FPT} cannot be translated to directed graphs. In particular, if a directed graph $G$ is 2-strongly-connected, it does not mean that for every two vertices $u$ and $v$, $G$ has a cycle containing  $u$ and $v$. We show the following theorem providing a full dichotomy for the complexity of \probLPDiam on $2$-strongly-connected graphs.

\begin{theorem}\label{thm:dichotomy_diameter}
	On $2$-strongly-connected directed graphs, \probLPDiam with $k\le 4$ can be solved in polynomial time, while for $k\ge 5$ it is \classNP-complete.
\end{theorem}

In what remains of this section, we prove the theorem. In Subsection~\ref{sec:alg-diam}, we show the positive part, and Subsection~\ref{sec:hard-daim} contains the hardness proof.

\subsection{Algorithm for $k\leq 4$}\label{sec:alg-diam}
We start with the positive part of Theorem~\ref{thm:dichotomy_diameter}. Note that it is sufficient to consider graphs with diameter greater than some fixed constant,
as in graphs with smaller diameter the problem can be solved in linear time.
The crucial part of the proof  is encapsulated in the following lemma, which states that a path of length $\diam(G)+4$ always exists in a $2$-strongly-connected graph $G$ of sufficiently large diameter. To construct such a path, we take the diameter pair $(s,t)$ and employ $2$-strong-connectivity of the graph to find two disjoint $(s,t)$-paths and two disjoint $(t,s)$-paths in the graph. We then show that out of the several possible ways to comprise a path out of the parts of these four paths, at least one always obtains a path of desired length.
The most non-trivial case of this construction involves constructing two paths of length five, one ending in a vertex $u$ that is at distance three from $s$ and the other starting in a vertex $v$ from which we can reach $t$ using three arcs.
We then concatencate these two paths using a specific $(u,v)$-path inbetween.
Since $(s,t)$ is a diameter pair, the length of any $(u,v)$-path is at least diameter minus six, so the length of the concatenation is at least diameter plus four. The other cases are analyzed in a similar fashion.

%%%%%%%%%%%%%%%%%%%%%%%%%%%%%%%%%%%%%%%%%%%%%%%%%%%%%%
\newcommand{\diamind}[1]{2^{3^{#1}}}
%%%%%%%%%%%%%%%%%%%%%%%%%%%%%%%%%%%%%%%%%%%%%%%%%%%%%%

\begin{lemma}\label{lemma:path_above_diameter}
	Any 2-strongly-connected directed graph $G$ with $\diam(G)\ge \diamind{17}$ has a path of length $\diam(G)+4$.
\end{lemma}
\begin{proof}

	Let $d=\diam(G)$ and $(s,t)$ be a pair such that $\diam(G)=\dist_G(s,t)$.
	Since $G$ is 2-strongly-connected, there exist two internally disjoint paths $P_1$ and $P_2$ from $s$ to $t$. 
	Denote the number of internal vertices of $P_1$ and $P_2$ by $p_1$ and $p_2$ respectively.
	Denote the internal vertices of $P_i$ by $v_{i,1}, v_{i,2},\ldots, v_{i,p_i}$ for each $i\in[2]$.
	Since $\diam(G)=\dist_G(s,t)$, we know that $p_i\ge d-1$. Therefore, if the length of $P_i$ is at least $d+4$, then $P_i$ is a path of length at least $\diam(G)+4$ and we are done. Hence, from now on we assume that $p_i \le d+2$.
	
	We say that a path between two arbitrary vertices in $G$ is an \emph{outer} path if no internal vertices belong to $V(P_1)\cup V(P_2)$.
	We now investigate sufficient conditions for $G$ to contain a path of length at least $d+4$.
	
	\begin{claim}\label{claim:diameter_no_long_alter}
		If there exists an outer path in $G$ going from $v_{i,j}$ to $v_{3-i,j'}$ with $j' \le j-3$, then there exists a path of length at least $d+4$ in $G$.
	\end{claim}
	\begin{proof}[Proof of Claim~\ref{claim:diameter_no_long_alter}]
		Let $T$ be such a path and consider the path $sP_iv_{i,j}Tv_{3-i,j'}P_{3-i}t$.
		This path is an $(s,t)$-path of length at least $j+1+(p_{3-i}+1-j')\ge d+(j-j')+1\ge d+4$.
	\end{proof}

	\begin{claim}\label{claim:diameter_longjump}
	If there exists an outer $(v_{i,j},s)$-path in $G$ with $i \in [2]$ and $j \ge 4$, then $G$ has a path of length at least $\diam(G)+4$.
	The same holds for an outer $(t,v_{i,j})$-path with $j \le p_i-3$.
\end{claim}
\begin{proof}[Proof of Claim~\ref{claim:diameter_longjump}]
	Assume that a $(v_{i,j},s)$-path with described properties exists.
	Then concatenate the path $v_{i,1}P_iv_{i,j}$, the $(v_{i,j},s)$-path and the path $P_{3-i}$.
	As all three paths are internally disjoint, we obtain a $(v_{i,1},t)$-path of length at least $(j-1)+1+d=d+j\ge d+4$ in $G$ as desired.
	
	The case of a $(t,v_{i,j})$-path is symmetrical and we need to concatenate the path $P_{3-i}$ with the $(t,v_{i,j})$-path and with the path $v_{i,j}P_iv_{i,p_i}$.
	The combined path is of length at least $d+1+(p_i-j)\ge d+4$.
\end{proof}

The following lemma shows that we can find either a path of length $d+4$ or many outer paths connecting $P_1$ and $P_2$ in $G$.
	\begin{lemma}\label{lemma:diameter_many_alter}
		If $G$ has no path of length at least $d+4$, then in any $(t,s)$-path and for every $i \in [2]$ there are at least $8$ outer subpaths going from an inner vertex of $P_i$ to an inner vertex of $P_{3-i}$.
	\end{lemma}
	\begin{proof}
		Take an $(t,s)$-path $Q$.
		If $Q$ has no inner vertices in $V(P_1)\cup V(P_2)$, then we can concatenate $v_{1,1}P_1t$ with $Q$ and then with $sP_2v_{2,p_2}$ and obtain a path of length at least $p_1+p_2+1>d+4$.
		
		Thus, $Q$ should have at least one inner vertex in $V(P_1)\cup V(P_2)$.
		Denote all inner vertices of $Q$ from $V(P_1)\cup V(P_2)$ by $q_1, q_2, \ldots, q_z$ in the order they appear on $Q$.
		Hence, $tQq_1$, $q_zQs$ and $q_kQq_{k+1}$ for every $k \in [z-1]$ are outer paths in $G$.
		
		Without loss of generality, we can assume that $q_1 \in V(P_1)$.
		Let $r$ be the largest number such that $q_k \in V(P_1)$ for each $k \in [r]$.
		First note that $r \le 3$, otherwise we can concatenate $P_2$ with $tQq_r$ and obtain a path of length at least $d+r\ge d+4$.
		Suppose now that the length of $q_kP_1t$ is greater than $3(r-1)$ for some $k \in [r]$.
		The vertices $q_1,q_2,\ldots,q_{k-1},q_{k+1},\ldots,q_r$ split this path into $r-1$ parts, and the length of one of these parts is at least four.
		Hence, for some $a\in\{q_1,q_2,\ldots,q_r\}$ and $b \in \{q_1,q_2,\ldots,q_r,t\}$ the length of $aP_1b$ is at least four and contains no inner vertex among $q_1,q_2,\ldots,q_r,t$.
		Then concatenate $P_2$ with $tQaP_1b$ without the vertex $b$.
		The obtained path is of length at least $d+1+3=d+4$.
		Thus, we have that for each $k \in [r]$ the vertex $q_k$ is at distance at most $3(r-1)\le 6$ from $t$ on $P_1$.
		In particular, $q_r=v_{1,j}$ for $j \ge p_1-5$.
		
		If $r=t$, i.e.\ there is no vertex from $P_2$ among $q_1,\ldots, q_r$, then we have an outer $(q_r,s)$-path in $G$.
		Since $q_r=v_{1,j}$ for $j\ge p_1-5\ge 4$, by Claim~\ref{claim:diameter_longjump} $G$ has a path of length at least $d+4$.
		We now have that $z>r$ and $q_{r+1}\in V(P_2)$, i.e.\ $Q$ alternates at least once between $P_1$ and $P_2$.
		We say that $k$ is an alternation point in $Q$ if $q_k \in V(P_i)$ and $q_{k+1}\in V(P_{3-i})$ for some $i \in [2]$.
		As a convenient exception, we also consider $k=z$ an alternation point in $Q$.
		Let $k_1<k_2<\ldots <k_{c-1}<k_c$ be the sequence of all such alternation points in $Q$.
		We know that $c\ge 1$ and $k_1=r$ and $k_c=z$.

		\begin{claim}\label{claim:diameter_close_to_t}
		For each $j \in [c]$, for every $k\in [k_j]$ the distance between $q_k$ and $t$ on $P_1$ or $P_2$ is at most $\diamind{j}$.
		\end{claim}
		\begin{proof}[Proof of Claim~\ref{claim:diameter_close_to_t}]
		Let us prove this lemma by induction, where the case $j=1$ has already been proved.
		Take $j>1$ and assume that the induction hypothesis holds for $j-1$.
		Then we have an outer $(q_{k_{j-1}},q_{k_{j-1}+1})$-path in $G$ and for each $k \in [k_{j-1}+1,k_j]$ $q_k\in V(P_i)$.
		We know that $q_{k_{j-1}}=v_{3-i,j'}$ and $j'\ge p_{3-i}-\diamind{j-1}+1$ by induction.
		If $k_j-k_{j-1}>\diamind{j-1}+3$, then $sP_{3-i}q_{k_{j-1}}Qq_{k_j}$ is a path of length at least $p_{3-i}+5\ge d+4$.
		Hence, we have that $k_j-k_{j-1} \le \diamind{j-1}+3$.
		On the other hand, $q_{k_{j-1}+1}=v_{i,j''}$ and by Claim~\ref{claim:diameter_no_long_alter} we have that $j''\ge j-2\ge p_{3-i}-\diamind{j-1}-1\ge p_i-3^{2^{j-1}}-4$.
		If $k_j=k_{j-1}+1$, then everything is proved since $v_{i,j''}$ is on the distance $p_i-j''+1\le\diamind{j-1}+5\le \diamind{j}$ from $t$ on $P_i$.
		From now we assume that $k_j-k_{j-1}>1$.
		
		Let $\ell,u\in  [k_{j-1}+1,k_j]$ be such that $q_\ell$ is the farthest from $t$ on $P_i$ and $q_u$ is the closest to $t$ on $P_i$.
		Then $q_\ell P_i q_u$ contains $q_k$ for each $k \in [k_{j-1}+1,k_j]$.
		The vertices $q_{k_{j-1}+1},\ldots, q_{k_j}$ split this path into $k_j-(k_{j-1}+1)$ parts.
		If the length of $q_\ell P_i q_u$ is more than $(\diamind{j-1}+2)(k_j-k_{j-1}-1)$, then one of these parts has length at least $\diamind{j-1}+3$.
		Denote the endpoints of this part by $q_a$ and $q_b$, i.e.\ this part is $q_aP_iq_b$.
		Then consider the path $sP_{3-i}q_{k_{j-1}}Qq_aP_iq_b$ without the vertex $q_b$.
		This path is of length at least $(p_{3-i}-\diamind{j-1}+1)+1+(\diamind{j-1}+3)\ge p_{3-i}+5\ge d+4$.
		Thus, we now assume that the length of $q_\ell P_iq_u$ is at most $(\diamind{j-1}+2)(k_j-k_{j-1}-1)\le (\diamind{j-1}+2)^2$.
		
		Then for each $k \in [k_{j-1}+1,k_j]$ the vertex $q_k$ is on a distance at most $(\diamind{j-1}+2)^2$ from $q_{k_{j-1}+1}$ on $P_i$.
		Since $q_{k_{j-1}+1}$ is on a distance at most $\diamind{j-1}+5$ from $t$, we have that for each such $q_k$ the distance between $q_k$ and $t$ on $P_i$ is at most \[\diamind{j-1}+5+(\diamind{j-1}+2)^2\le 2^{2\cdot 3^{j-1}}+5\cdot \diamind{j-1}+9\le3\cdot 2^{2\cdot 3^{j-1}}\le \diamind{j}.\]
		\end{proof}
	
		Using Claim~\ref{claim:diameter_close_to_t}, we obtain that $q_z=q_{k_c}=v_{i,j}$, where $i\in[2]$ and $j \ge p_i-\diamind{c}+1$.
		$Q$ also yields an outer $(q_z,s)$-path, and by Claim~\ref{claim:diameter_longjump}, we have that $j\le 3$ or $G$ has a path of length $d+4$.
		Then $p_i\le \diamind{c}+2$.
		Since $p_i \ge d-1$, we obtain that $\diamind{c}+3\ge d$ so $c\ge  17$.
		It follows that for each $i\in[2]$ there are at least $8$ outer $(v_{i,j},v_{3-i,j'})$-paths in $G$, and all of them are subpaths of $Q$.
	\end{proof}

	The last claim in the proof shows that we can find two disjoint paths, one near $s$ and one near $t$ in $G$.
	We shall then combine them in a single path of length $d+4$ in $G$ using paths from Lemma~\ref{lemma:diameter_many_alter}.
	
	\begin{claim}\label{claim:diameter_paths}
		In $G$, there is either:
		\begin{enumerate}
			\item A path of length at least $d+4$, or
			\item A path of length $5$ ending in $v_{1,3}$ or in $v_{2,3}$ that avoids all vertices of form $v_{i,j}$ for $j>3$, and
			
			a path of length $5$ starting in $v_{1,p_1-2}$ or in $v_{2,p_2-2}$ that avoids all vertices of form $v_{i,j}$ for $j<p_2-2$.
			These two paths do not share any common vertex.
		\end{enumerate}
	\end{claim}
	\begin{proof}[Proof of Claim~\ref{claim:diameter_paths}]
		Since $G$ is 2-strongly-connected, there are two internally disjoint $(s,t)$-paths $Q_1$ and $Q_2$.
		By Lemma~\ref{lemma:diameter_many_alter}, either $G$ contains a path of length $d+4$ or for each $i\in [2]$ $Q_i$ contains at least four vertices in $V(P_1)\cup V(P_2)\setminus \{s,t\}$.
		Hence, for each $k \in [2]$, we have four outer paths $tQ_ka_k$, $a_kQ_kb_k$, $d_kQ_ks$, $c_kQ_kd_k$ in $G$.
		Note that all these eight paths are internally disjoint and the eight vertices $a_k,b_k,c_k,d_k$ are pairwise distinct.
		
		We first show how to construct a path of length $5$ ending in $v_{i,3}$ for some $i\in[2]$, using the paths $d_kQ_ks$ and $c_kQ_kd_k$.
		If $d_k=v_{i,j}$ for $k,i\in [2]$ and $j \ge 2$, then take the path $v_{i,1}P_iv_{i,j}Q_ksP_{3-i}v_{3-i,3}$.
		This is a path of length at least $(j-1)+1+3\ge j+3\ge 5$ ending in $v_{3-i,3}$ avoiding all vertices of form $v_{x,y}$ with $y>3$ as required.
		Hence, it is left to consider the case when for each $k\in [2]$ $d_k=v_{i,j}$ where $j\le 2$.
		Without loss of generality, we assume that for each $k\in [2]$ $d_k=v_{k,1}$.
		
		We now require the path $c_1Q_1d_1$ for the construction.
		Let $c_1=v_{i',j'}$ for some $i'\in [2]$ and $j'\in [p_{i'}]$.
		Suppose first that $i'=1$.
		If $j'\ge 4$ then take the path $v_{1,2}P_1v_{1,j'}Q_1v_{1,1}Q_1sP_{2}t$.
		This is a path of length at least $(j'-2)+2+d\ge d+j'\ge d+4$ in $G$.
		If $j' \le 3$, then consider the path $v_{1,j'}Q_1v_{1,1}Q_1sP_{2}v_{2,3}$.
		This is a path of length at least five ending in $v_{2,3}$ that avoids all vertices $v_{x,y}$ with $y>3$.
		
		It is left to consider $c_1=v_{2,j'}$.
		If $j'\ge 4$, then $c_1Q_1d_1$ is an outer $(v_{2,j'},v_{1,j})$-path with $j\le j'-3$ and we are done by Claim~\ref{claim:diameter_no_long_alter}.
		Hence, $j'\in \{2,3\}$.
		Then consider $sP_{2}v_{2,j'}Q_1v_{1,1}P_1v_{1,3}$.
		This path has length at least $j'+1+2\ge 5$, ends in $v_{1,3}$ and avoids all vertices required.
		The part of the proof for paths ending in $v_{i,3}$ is complete.
		Note that all constructed paths can contain only vertices $v_{i,j}$ with $j\le 3$, the vertex $s$ and inner vertices of the paths $c_kQ_ks$.

		The proof for paths starting in $v_{i,p_i-2}$ is symmetrical.
		The symmetry lies in that when we take the transpose of $G$, the roles of $s$ and $t$ exchange, and paths starting in $v_{i,p_i-2}$ become paths ending in $v_{i,3}$.
		Role of the paths $c_kQ_ks$ take the transpose of the paths $tQ_kb_k$.
		Hence, we can consider the graph $G^T$, exchange $s$ and $t$ and renumerate each vertex $v_{i,j} := v'_{i,p_i-j+1}$ for $i \in [2]$ and $j \in [p_i]$. By applying the previous proof to $G^T$, we obtain either a path of length $d+4$ in $G^T$ or a path of length at least $5$ ending in $v'_{i,3}$.
		Thus in $G$ we obtain either a path of length at least $d+4$ or a path of length at least $5$ ending in $v_{i,p_i-2}$ and avoiding all required vertices.		
		Note that the constructed paths in this part can contain only vertices $v_{i,j}$ with $j\ge p_{i}-2$, the vertex $t$ and inner vertices of the paths $tQ_kb_k$.
		That is, the two paths of length five from the different parts of the proof do not share any common vertex.
	\end{proof}

	To conclude the proof of Lemma~\ref{lemma:path_above_diameter}, we combine results of Lemma~\ref{lemma:diameter_many_alter} and Claim~\ref{claim:diameter_paths} together.
	By Claim~\ref{claim:diameter_paths}, if $G$ doesn't contain a path of length $d+4$, there exists a path of length five ending in $v_{i,3}$ avoiding all $v_{x,y}$ with $y\ge 4$ and a path of length five starting in $v_{i',p_{i'}-2}$ avoiding all $v_{x,y}$ with $y\le p_{i'}-3$.
	Denote these paths by $R$ and $R'$ respectively.
	If $i=i'$, then we take the concatenation $R\circ v_{i,3}P_i v_{i,p_i-2}\circ R'$.
	This is indeed a path without self-intersections as $R$ and $R'$ avoid all vertices of $P_i$ between $v_{i,3}$ and $v_{i,p_i-2}$ and are disjoint.
	The obtained path is of length $5+(p_i-2-3)+5\ge p_i+5\ge d+4$.
	
	If $i'=3-i$, then we require an outer $(v_{i,y},v_{i',y'})$-path $T$ with $y\ge 3$ and $y'\le p_{i'}-2$ for the concatenation $R\circ v_{i,3}P_iv_{i,y}Tv_{i',y'}P_{i'}v_{i',p_{i'}-2}\circ R'$.
	This path can only share vertices $v_{i,3}$ and $v_{i',p_{i'}-2}$ with $R$ and $R'$.
	Thus, we want $T$ to avoid vertices $v_{i,1},v_{i,2},v_{i',p_{i'}},v_{i',p_{i'-1}}$ and all vertices in $V(R)\cup V(R')\setminus \{v_{i,3},v_{i',p_{i'}-2}\}$.
	These sum up to a total of $14$ vertices that should be avoided.
	Since there are two internally disjoint $(t,s)$-paths in $G$, by Lemma~\ref{lemma:diameter_many_alter} there are at least $16$ outer paths in $G$ going from an inner vertex of $P_i$ to an inner vertex of $P_{i'}$.
	As each vertex to avoid lies on at most one paths among these $16$, at least two paths are suitable candidates for $T$.
	Take any of these candidates and denote it by $T$.
	
	To estimate the length of $T$, consider the path $sP_iv_{i,y} T v_{i',y'}P_{i'}t$.
	The length of this concatenation equals $y+\ell+(p_{i'}-y'+1)=(p_{i'}+1)-(y'-y)+\ell$, where $\ell$ is the length of $T$.
	Since the concatention is an $(s,t)$-path we have $(p_{i'}+1)-(y'-y)+\ell\ge d$, so the length of $T$ is at least $(y'-y)-(p_{i'}+1-d)$.
	The length of the path $v_{i,3}P_iv_{i,y}Tv_{i',y'}P_{i'}v_{i',p_{i'}-2}$ is at least $(y-3)+(y'-y)-(p_{i'}+1-d)+(p_{i'}-2-y')=d-6$.
	It follows that $R\circ v_{i,3}P_iv_{i,y}Tv_{i',y'}P_{i'}v_{i',p_{i'}-2}\circ R'$ is of length at least $d+4$ as required.
\end{proof}

	We note that the proof of Lemma~\ref{lemma:path_above_diameter} is constructive and can be turned into a polynomial-time algorithm finding a path of length $\diam(G)+4$ in a graph with diameter at least $\diamind{17}$.
	For turning the proof into an algorithm, we require a procedure to find two internally disjoint $(s,t)$-paths or two internally disjoint $(t,s)$-paths in $G$.
	This can be done in polynomial time using any polynomial-time maximum flow algorithm.
	For $\diam(G)<\diamind{17}$, we use the color coding algorithm for \textsc{Longest Path} to find a path of constant length $\diam(G)+4$.
	The running time of this algorithm is linear in $n$.
	We obtain that \probLPDiam~with $k\le 4$ can be solved in polynomial time on 2-strongly-connected digraphs.

\subsection{NP-Hardness}\label{sec:hard-daim}
We proceed to the second and negative result of Theorem~\ref{thm:dichotomy_diameter}.
The general idea  of the proof is similar to that of  Proposition~\ref{prop:diam-hard}.
We aim to take a path-like gadget graph, then take a sufficiently large \textsc{Hamiltonian Path} instance and connect it to the middle of the gadget.
However, while in the general case it suffices to simply take a path graph (Proposition~\ref{prop:diam-hard}), the $2$-strongly-connected case is much more technically involved.
First, we need a family of gadget graphs that are 2-strongly-connected, have arbitrarily large diameter, but each graph in the family does not have a path longer than diameter plus four.
This, in fact, is exactly a counterexample to the positive part of Theorem~\ref{thm:dichotomy_diameter},
as the existence of such family of graphs proves that there cannot always be a path of length diameter plus four in a sufficiently large $2$-connected directed graph.
Additionally, for the reduction we need that graphs in this family behave like paths,
specifically that the length of the longest path that ends in the ``middle'' of the gadget is roughly half of the diameter.
Constructing this graph family is a main technical challenge of the theorem.
After constructing the gadget graph family the proof is reasonably simple, as we take a 2-connected \textsc{Hamiltonian Path} instance, and connect it to the ``middle'' of a sufficiently large gadget graph.
The connection is done by a simple 4-vertex connector gadget that ensures that the resulting graph is 2-strongly-connected, but only allows for paths that alternate at most once between the gadget graph and the starting instance.
The whole reduction is visualized in Figure~\ref{fig:diameter_reduction}.

We start the proof with a construction of a family of directed graphs ${G_1,G_2,\ldots,G_\ell,\ldots}$ that are 2-strongly-connected, while the longest path in  ${G_\ell}$ has length $\diam({G_\ell})+4$ starting from some $\ell$.
We shall afterwards use this graph for a many-to-one reduction from \textsc{Hamiltonian Path} to \probLPDiam.

\medskip
\noindent \textbf{Construction of the graph ${G_\ell}$.}
We construct $G_\ell$ for arbitrary $\ell\ge 1$.
We require three types of gadgets for the construction.
The first two are the source and the sink gadgets, presented in Figure~\ref{fig:diameter}a and Figure~\ref{fig:diameter}c.
They contain vertices $s$ and $t$ respectively.
Note that the sink gadget is isomorphic to the transpose of the source gadget, and isomorphism is clear from enumeration of the vertices of both gadgets.
The third type of gadget, namely the hat gadget, is presented in Figure~\ref{fig:diameter}b, and consists of ten vertices.
To construct the graph $G_\ell$ for $\ell\ge 1$, we take one source gadget, $2\ell-1$ hat gadgets, and one sink gadget.
Then we identify the vertices $s_8$ and $s_{14}$ with the vertices $h_1$ and $h_4$ of the first hat gadget respectively.
Further, for each $i \in [2\ell-2]$, we identify the vertices $h_3$ and $h_{10}$ of the $i^\text{th}$ hat gadget respectively with the vertices $h_4$ and $h_1$ of the $(i+1)^\text{th}$ hat gadget.
Finally, we identify the vertices $h_3$ and $h_{10}$ of the last, $(2\ell-1)^\text{th}$ gadget, with the vertices $t_8$ and $t_{14}$ of the sink gadget.
Thus, all gadgets are arranged into a chain and form a weakly connected graph.
This is the graph $G_\ell$, and later in this section we prove that it is 2-strongly-connected.

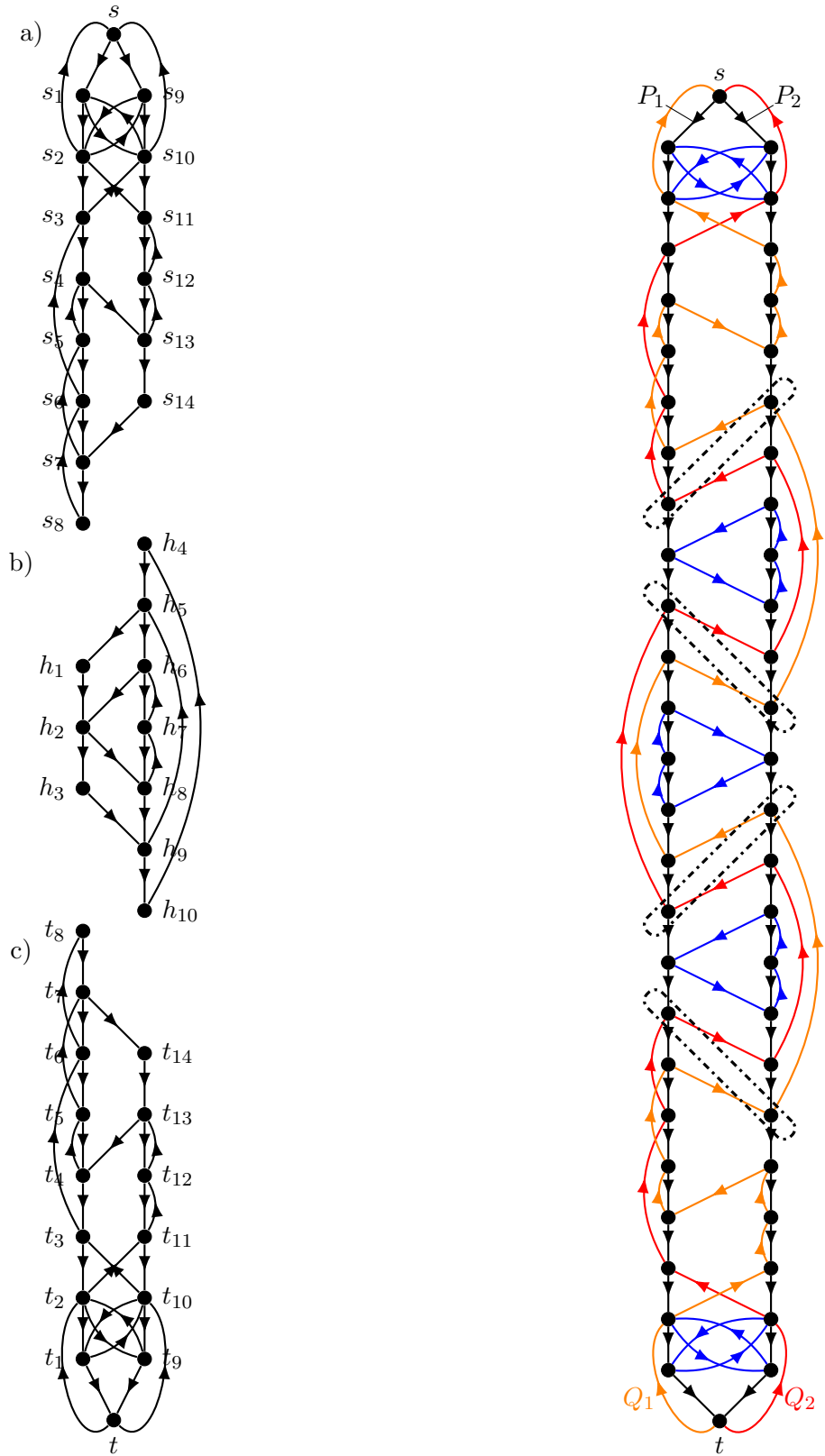
\begin{figure}	
	\centering
	\begin{tikzpicture}[scale=0.3,x={(0cm,-1cm)},y={(-1cm,0cm)}]
	\tikzstyle{vertex}=[circle, fill=black,draw,inner sep=0,minimum size=0.2cm]
	\usetikzlibrary{decorations.markings}
	\usetikzlibrary{arrows.meta}
	\tikzstyle{medge}=[thick,decoration={
		markings,
		mark=at position 0.6 with {\arrow{Latex}}},postaction={decorate}]
	\tikzstyle{redpath}=[]
	\tikzstyle{orangepath}=[]
	\tikzstyle{purpleedges}=[]
	\usetikzlibrary{math} 
	\usetikzlibrary{patterns}

	\def\n{25};
	\def\nm{24};
	\def\gap{3}

	% pattern
	\newcommand{\drawPatternUB}[2]
	{
		\pgfmathtruncatemacro{\r}{#1};
		\pgfmathtruncatemacro{\o}{#2};
		
		\pgfmathtruncatemacro{\nr}{\r-1};
		\pgfmathtruncatemacro{\no}{\o-1};
		
		\draw[redpath,medge] (v-\nr) to (v-\r);
		\draw[orangepath,medge] (v-\no) to (v-\o);
		
		\pgfmathtruncatemacro{\r}{\nr};
		\pgfmathtruncatemacro{\o}{\no};
		
		\pgfmathtruncatemacro{\nr}{\r + 8};
		\pgfmathtruncatemacro{\no}{\o + 12};
		
		\draw[redpath,medge] (v-\nr) to[ bend right] (v-\r);
		\draw[orangepath,medge] (v-\no) to[ bend right] (v-\o);
		
		\pgfmathtruncatemacro{\r}{\nr};
		\pgfmathtruncatemacro{\o}{\no};
		
		\pgfmathtruncatemacro{\a}{\r - 2};
		\pgfmathtruncatemacro{\b}{\a - 2};
		\pgfmathtruncatemacro{\c}{\b - 2};
		\pgfmathtruncatemacro{\d}{\c + 1};
		
		\draw[purpleedges,medge] (v-\a) to [ bend right] (v-\b);
		\draw[purpleedges,medge] (v-\b) to [ bend right] (v-\c);
		\draw[purpleedges,medge] (v-\c) to (v-\d);
		\draw[purpleedges,medge] (v-\d) to (v-\a);
	};
	
	\newcommand{\drawPatternBU}[2]
	{
		\pgfmathtruncatemacro{\r}{#1};
		\pgfmathtruncatemacro{\o}{#2};
		
		\pgfmathtruncatemacro{\nr}{\r-3};
		\pgfmathtruncatemacro{\no}{\o-3};
		
		\draw[redpath,medge] (v-\nr) to (v-\r);
		\draw[orangepath,medge] (v-\no) to (v-\o);
		
		\pgfmathtruncatemacro{\r}{\nr};
		\pgfmathtruncatemacro{\o}{\no};
		
		\pgfmathtruncatemacro{\nr}{\r + 12};
		\pgfmathtruncatemacro{\no}{\o + 8};
		
		\draw[redpath,medge] (v-\nr) to[ bend left] (v-\r);
		\draw[orangepath,medge] (v-\no) to[ bend left] (v-\o);
		
		\pgfmathtruncatemacro{\r}{\nr};
		\pgfmathtruncatemacro{\o}{\no};
		
		\pgfmathtruncatemacro{\a}{\o - 2};
		\pgfmathtruncatemacro{\b}{\a - 2};
		\pgfmathtruncatemacro{\c}{\b - 2};
		\pgfmathtruncatemacro{\d}{\c + 3};
		
		\draw[purpleedges,medge] (v-\a) to [ bend left] (v-\b);
		\draw[purpleedges,medge] (v-\b) to [ bend left] (v-\c);
		\draw[purpleedges,medge] (v-\c) to (v-\d);
		\draw[purpleedges,medge] (v-\d) to (v-\a);
	};

	\def\tikzsegment#1#2#3{ % This is the macro explained above
		\path let
		\p1=($(#3)-(#2)$),
		\n1={veclen(\p1)*1.25}
		in (#2) -- (#3) 
		node[minimum width=\n1, 
		inner sep=0pt, 
		pos=0.5,sloped,rectangle,rounded corners,
		#1] 
		(line){};
	}
	
	\newcommand{\drawCut}[2] {
		\tikzsegment{very thick,dash dot,draw,minimum height=1cm}{v-#1}{v-#2};
	}

	\begin{scope}[shift={(17,0)}]
		\foreach \x in {6,7,8,9,10,11,12} {
			
			\pgfmathtruncatemacro{\label}{\x*2+1}
			\pgfmathtruncatemacro{\l}{\x-2}
			\node [vertex,label=right:{$h_{\l}$}] (v-\label) at (-4+\gap*\x,-1.5) {};
		}
		\foreach \x in {8,9,10} {
			\pgfmathtruncatemacro{\label}{\x*2}
			\pgfmathtruncatemacro{\l}{\x-7}
			\node [vertex,label=left:{$h_{\l}$}] (v-\label) at (-4+\gap*\x,1.5) {};
			
		}
	
	\node at ($(v-13)+(1,6)$) {b)};
		
		\foreach \x in {8,9} {
			\pgfmathtruncatemacro{\labela}{\x*2}
			\pgfmathtruncatemacro{\labelb}{\x*2+2}
			\draw[medge] (v-\labela) -- (v-\labelb);
			
		}
		\foreach \x in {6,...,11} {
			\pgfmathtruncatemacro{\labela}{\x*2+1}
			\pgfmathtruncatemacro{\labelb}{\x*2+3}
			\draw[medge] (v-\labela) -- (v-\labelb);
			
		}
		\pgfmathtruncatemacro{\r}{16};
		\pgfmathtruncatemacro{\o}{14};
		
		\pgfmathtruncatemacro{\nr}{\r-1};
		\pgfmathtruncatemacro{\no}{\o-1};
		
		\draw[redpath,medge] (v-\nr) to (v-\r);
		%\draw[orangepath,medge] (v-\no) to (v-\o);
		
		\pgfmathtruncatemacro{\r}{\nr};
		\pgfmathtruncatemacro{\o}{\no};
		
		\pgfmathtruncatemacro{\nr}{\r + 8};
		\pgfmathtruncatemacro{\no}{\o + 12};
		
		\draw[redpath,medge] (v-\nr) to[ bend right] (v-\r);
		\draw[orangepath,medge] (v-\no) to[ bend right] (v-\o);
		
		\pgfmathtruncatemacro{\r}{\nr};
		\pgfmathtruncatemacro{\o}{\no};
		
		\pgfmathtruncatemacro{\nr}{\r -3};
		\draw[redpath,medge] (v-\nr) to (v-\r);
		
		\pgfmathtruncatemacro{\a}{\r - 2};
		\pgfmathtruncatemacro{\b}{\a - 2};
		\pgfmathtruncatemacro{\c}{\b - 2};
		\pgfmathtruncatemacro{\d}{\c + 1};
		
		\draw[purpleedges,medge] (v-\a) to [ bend right] (v-\b);
		\draw[purpleedges,medge] (v-\b) to [ bend right] (v-\c);
		\draw[purpleedges,medge] (v-\c) to (v-\d);
		\draw[purpleedges,medge] (v-\d) to (v-\a);
	\end{scope}
	
	\begin{scope}[shift={(10,0)}]
		\node at (-4, 4) {a)};
		\node [vertex,label=above:{$s$}] (s) at (-4,0) {};
		\foreach \x in {1,2,3,4,5,6} {
			
			\pgfmathtruncatemacro{\label}{\x*2+1}
			\pgfmathtruncatemacro{\l}{\x+8}
			\node [vertex,label=right:{$s_{\l}$}] (v-\label) at (-4+\gap*\x,-1.5) {};
		}
		\foreach \x in {1,2,3,4,5,6,7,8} {
			\pgfmathtruncatemacro{\label}{\x*2}
			\pgfmathtruncatemacro{\l}{\x}
			\node [vertex,label=left:{$s_{\l}$}] (v-\label) at (-4+\gap*\x,1.5) {};
			
		}
		
		\foreach \x in {1,...,7} {
			\pgfmathtruncatemacro{\labela}{\x*2}
			\pgfmathtruncatemacro{\labelb}{\x*2+2}
			\draw[medge] (v-\labela) -- (v-\labelb);
			
		}
		\foreach \x in {1,...,5} {
			\pgfmathtruncatemacro{\labela}{\x*2+1}
			\pgfmathtruncatemacro{\labelb}{\x*2+3}
			\draw[medge] (v-\labela) -- (v-\labelb);
			
		}
		\draw[medge] (s) -- (v-2);
		\draw[medge] (s) -- (v-3);
	% red path here
	\draw[redpath,medge] (v-16) to [ bend left] (v-12);
	\draw[redpath,medge] (v-12) to [ bend left] (v-6);
	\draw[redpath,medge] (v-6) to (v-5);
	\draw[redpath,medge] (v-5) to[out=45,in=45] (s);
	
	% orange
	\draw[orangepath,medge] (v-14) to [ bend left] (v-10);
	\draw[orangepath,medge] (v-10) to [ bend left] (v-8);
	\draw[orangepath,medge] (v-8) to (v-11);
	\draw[orangepath,medge] (v-11) to[ bend right] (v-9);
	\draw[orangepath,medge] (v-9) to[ bend right] (v-7);
	\draw[orangepath,medge] (v-7) to (v-4);
	\draw[orangepath,medge] (v-4) to[out=135,in=135] (s);
	
	% purple
	\draw[purpleedges,medge] (v-3) to[ bend right] (v-4);
	\draw[purpleedges,medge] (v-4) to[ bend right] (v-3);
	\draw[purpleedges,medge] (v-2) to[ bend right] (v-5);
	\draw[purpleedges,medge] (v-5) to[ bend right] (v-2);
	\draw[medge] (v-13) -- (v-14);
	\end{scope}
	
	\begin{scope}[shift={(0,0)}]
		\node [vertex,label=below:{$t$}] (t) at (-4+\gap*\n+\gap,0) {};
		\foreach \x in {20,...,25} {
			
			\pgfmathtruncatemacro{\label}{\x*2+1}
			\pgfmathtruncatemacro{\l}{26-\x+8}
			\node [vertex,label=right:{$t_{\l}$}] (v-\label) at (-4+\gap*\x,-1.5) {};
		}

		\foreach \x in {18,...,25} {
			\pgfmathtruncatemacro{\label}{\x*2}
			\pgfmathtruncatemacro{\l}{26-\x}
			\node [vertex,label=left:{$t_{\l}$}] (v-\label) at (-4+\gap*\x,1.5) {};
			
		}
	\node at ($(v-36)+(1,3)$) {c)};		
				\foreach \x in {18,...,24} {
			\pgfmathtruncatemacro{\labela}{\x*2}
			\pgfmathtruncatemacro{\labelb}{\x*2+2}
			\draw[medge] (v-\labela) -- (v-\labelb);
			
		}
		\foreach \x in {20,...,24} {
			\pgfmathtruncatemacro{\labela}{\x*2+1}
			\pgfmathtruncatemacro{\labelb}{\x*2+3}
			\draw[medge] (v-\labela) -- (v-\labelb);
			
		}
		\draw[medge] (v-50) -- (t);
		\draw[medge] (v-51) -- (t);
	\pgfmathtruncatemacro{\r}{(\n-7)*2}
	\pgfmathtruncatemacro{\nr}{\r+3}

	\pgfmathtruncatemacro{\nr}{\r+4}
	
	% red path here
	\draw[redpath,medge] (v-\nr) to [ bend left] (v-\r);
	\pgfmathtruncatemacro{\r}{\nr}
	\pgfmathtruncatemacro{\nr}{\r+6}
	\draw[redpath,medge] (v-\nr) to [ bend left] (v-\r);
	\pgfmathtruncatemacro{\r}{\nr}
	\pgfmathtruncatemacro{\nr}{\r+3}
	\draw[redpath,medge] (v-\nr) to (v-\r);
	\draw[redpath,medge] (t) to[out=315,in=315] (v-\nr);
	
	% orange
	\pgfmathtruncatemacro{\o}{(\n-6)*2}
	\pgfmathtruncatemacro{\no}{\o+3}
	\draw[orangepath,medge] (v-\o) to  (v-\no);
	\pgfmathtruncatemacro{\no}{\o+4}
	\draw[orangepath,medge] (v-\no) to [ bend left] (v-\o);
	\pgfmathtruncatemacro{\o}{\no}
	\pgfmathtruncatemacro{\no}{\o+2}
	\draw[orangepath,medge] (v-\no) to [ bend left] (v-\o);
	\pgfmathtruncatemacro{\o}{\no}
	\pgfmathtruncatemacro{\no}{\o-1}
	\draw[orangepath,medge] (v-\no) to (v-\o);
	\pgfmathtruncatemacro{\o}{\no}
	\pgfmathtruncatemacro{\no}{\o+2}
	\draw[orangepath,medge] (v-\no) to [ bend right] (v-\o);
	\pgfmathtruncatemacro{\o}{\no}
	\pgfmathtruncatemacro{\no}{\o+2}
	\draw[orangepath,medge] (v-\no) to [ bend right] (v-\o);
	\pgfmathtruncatemacro{\o}{\no}
	\pgfmathtruncatemacro{\no}{\o+1}
	\draw[orangepath,medge] (v-\no) to (v-\o);
	\draw[orangepath,medge] (t) to[out=225,in=225] (v-\no);
	
	% purple
	\pgfmathtruncatemacro{\o}{\no}
	\pgfmathtruncatemacro{\no}{\o+3}
	\draw[purpleedges,medge] (v-\no) to[ bend right] (v-\o);
	\draw[purpleedges,medge] (v-\o) to[ bend right] (v-\no);
	\pgfmathtruncatemacro{\r}{\nr}
	\pgfmathtruncatemacro{\nr}{\r+1}
	\draw[purpleedges,medge] (v-\nr) to[ bend right] (v-\r);
	\draw[purpleedges,medge] (v-\r) to[ bend right] (v-\nr);
	\end{scope}
\end{tikzpicture}
	\hspace{150pt}
	\begin{tikzpicture}[scale=0.5,x={(0cm,-1cm)},y={(-1cm,0cm)}]
	\tikzstyle{vertex}=[circle, fill=black,draw,inner sep=0,minimum size=0.2cm]
	\usetikzlibrary{decorations.markings}
	\usetikzlibrary{arrows.meta}
	\tikzstyle{medge}=[thick,decoration={
		markings,
		mark=at position 0.6 with {\arrow{Latex}}},postaction={decorate}]
	\tikzstyle{medge5}=[thick,decoration={
		markings,
		mark=at position 0.5 with {\arrow{Latex}}},postaction={decorate}]
	\tikzstyle{medge7}=[thick,decoration={
		markings,
		mark=at position 0.7 with {\arrow{Latex}}},postaction={decorate}]
	\tikzstyle{medge8}=[thick,decoration={
		markings,
		mark=at position 0.8 with {\arrow{Latex}}},postaction={decorate}]
	\tikzstyle{redpath}=[red]
	\tikzstyle{orangepath}=[orange]
	\tikzstyle{purpleedges}=[blue]
	\usetikzlibrary{math} 
	\usetikzlibrary{patterns}
	\usetikzlibrary{quotes}
	
	\node [vertex,label=above:{$s$}] (s) at (-4,0) {};
	\node [inner sep=0,minimum size=0cm] (pp1) at (-4,2) {$P_1$};
	\node [inner sep=0,minimum size=0cm] (pp2) at (-4,-2) {$P_2$};	
	
	\def\n{25};
	\def\nm{24};
	\def\gap{1.5}

	\foreach \x in {1,...,\n} {
		\pgfmathtruncatemacro{\label}{\x*2}
		\node [vertex] (v-\label) at (-4+\gap*\x,1.5) {};
		\pgfmathtruncatemacro{\label}{\x*2+1}
		\node [vertex] (v-\label) at (-4+\gap*\x,-1.5) {};
	}
	
	\draw[medge] (s) -- (v-2) node[midway,inner sep=0,minimum size=0cm] (p1) {};
	\draw[medge] (s) -- (v-3) node[midway,inner sep=0,minimum size=0cm] (p2) {};

%	\node at ($(s)+(0,4)$) {e)};
	
	\foreach \x in {1,...,\nm} {
		\pgfmathtruncatemacro{\labela}{\x*2}
		\pgfmathtruncatemacro{\labelb}{\x*2+2}
		\draw[medge] (v-\labela) -- (v-\labelb);
		
		\pgfmathtruncatemacro{\labela}{\x*2+1}
		\pgfmathtruncatemacro{\labelb}{\x*2+3}
		\draw[medge] (v-\labela) -- (v-\labelb);
	}

	% red path here
	\draw[redpath,medge] (v-16) to [bend left] (v-12);
	\draw[redpath,medge] (v-12) to [bend left] (v-6);
	\draw[redpath,medge8] (v-6) to (v-5);
	\draw[redpath,medge] (v-5) to[out=45,in=45] (s);
	
	% orange
	\draw[orangepath,medge] (v-14) to [bend left] (v-10);
	\draw[orangepath,medge] (v-10) to [bend left] (v-8);
	\draw[orangepath,medge] (v-8) to (v-11);
	\draw[orangepath,medge] (v-11) to[ bend right] (v-9);
	\draw[orangepath,medge] (v-9) to[ bend right] (v-7);
	\draw[orangepath,medge8] (v-7) to (v-4);
	\draw[orangepath,medge] (v-4) to[out=135,in=135] (s);
	
	% purple
	\draw[purpleedges,medge7] (v-3) to[ bend right] (v-4);
	\draw[purpleedges,medge7] (v-4) to[ bend right] (v-3);
	\draw[purpleedges,medge5] (v-2) to[ bend right] (v-5);
	\draw[purpleedges,medge5] (v-5) to[ bend right] (v-2);
	
	% pattern
	\newcommand{\drawPatternUB}[2]
	{
		\pgfmathtruncatemacro{\r}{#1};
		\pgfmathtruncatemacro{\o}{#2};
		
		\pgfmathtruncatemacro{\nr}{\r-1};
		\pgfmathtruncatemacro{\no}{\o-1};
		
		\draw[redpath,medge] (v-\nr) to (v-\r);
		\draw[orangepath,medge] (v-\no) to (v-\o);
		
		\pgfmathtruncatemacro{\r}{\nr};
		\pgfmathtruncatemacro{\o}{\no};
		
		\pgfmathtruncatemacro{\nr}{\r + 8};
		\pgfmathtruncatemacro{\no}{\o + 12};
		
		\draw[redpath,medge] (v-\nr) to[ bend right] (v-\r);
		\draw[orangepath,medge] (v-\no) to[ bend right] (v-\o);
		
		\pgfmathtruncatemacro{\r}{\nr};
		\pgfmathtruncatemacro{\o}{\no};
		
		\pgfmathtruncatemacro{\a}{\r - 2};
		\pgfmathtruncatemacro{\b}{\a - 2};
		\pgfmathtruncatemacro{\c}{\b - 2};
		\pgfmathtruncatemacro{\d}{\c + 1};
		
		\draw[purpleedges,medge] (v-\a) to [ bend right] (v-\b);
		\draw[purpleedges,medge] (v-\b) to [ bend right] (v-\c);
		\draw[purpleedges,medge] (v-\c) to (v-\d);
		\draw[purpleedges,medge] (v-\d) to (v-\a);
	};
	
	\newcommand{\drawPatternBU}[2]
	{
		\pgfmathtruncatemacro{\r}{#1};
		\pgfmathtruncatemacro{\o}{#2};
		
		\pgfmathtruncatemacro{\nr}{\r-3};
		\pgfmathtruncatemacro{\no}{\o-3};
		
		\draw[redpath,medge] (v-\nr) to (v-\r);
		\draw[orangepath,medge] (v-\no) to (v-\o);
		
		\pgfmathtruncatemacro{\r}{\nr};
		\pgfmathtruncatemacro{\o}{\no};
		
		\pgfmathtruncatemacro{\nr}{\r + 12};
		\pgfmathtruncatemacro{\no}{\o + 8};
		
		\draw[redpath,medge] (v-\nr) to[bend left] (v-\r);
		\draw[orangepath,medge] (v-\no) to[bend left] (v-\o);
		
		\pgfmathtruncatemacro{\r}{\nr};
		\pgfmathtruncatemacro{\o}{\no};
		
		\pgfmathtruncatemacro{\a}{\o - 2};
		\pgfmathtruncatemacro{\b}{\a - 2};
		\pgfmathtruncatemacro{\c}{\b - 2};
		\pgfmathtruncatemacro{\d}{\c + 3};
		
		\draw[purpleedges,medge] (v-\a) to [bend left] (v-\b);
		\draw[purpleedges,medge] (v-\b) to [bend left] (v-\c);
		\draw[purpleedges,medge] (v-\c) to (v-\d);
		\draw[purpleedges,medge] (v-\d) to (v-\a);
	};

	\drawPatternUB{16}{14}
	
	\drawPatternBU{23}{25}
	
	\drawPatternUB{32}{30}
	
	\node [vertex,label=below:{$t$}] (t) at (-4+\gap*\n+\gap,0) {};
	\pgfmathtruncatemacro{\label}{\n*2}
	\draw[medge] (v-\label) to (t);
	\pgfmathtruncatemacro{\label}{\n*2+1};
	\draw[medge] (v-\label) to (t);
	
	\pgfmathtruncatemacro{\r}{(\n-7)*2}
	\pgfmathtruncatemacro{\nr}{\r+3}
	\draw[redpath,medge] (v-\r) to (v-\nr);
	\pgfmathtruncatemacro{\nr}{\r+4}

	% red path here
	\draw[redpath,medge] (v-\nr) to [bend left] (v-\r);
	\pgfmathtruncatemacro{\r}{\nr}
	\pgfmathtruncatemacro{\nr}{\r+6}
	\draw[redpath,medge] (v-\nr) to [bend left] (v-\r);
	\pgfmathtruncatemacro{\r}{\nr}
	\pgfmathtruncatemacro{\nr}{\r+3}
	\draw[redpath,medge8] (v-\nr) to (v-\r);
	\draw[redpath,medge] (t) to[out=315,in=315,"$Q_2$" right] (v-\nr);
	
	% orange
	\pgfmathtruncatemacro{\o}{(\n-6)*2}
	\pgfmathtruncatemacro{\no}{\o+3}
	\draw[orangepath,medge] (v-\o) to  (v-\no);
	\pgfmathtruncatemacro{\no}{\o+4}
	\draw[orangepath,medge] (v-\no) to [bend left] (v-\o);
	\pgfmathtruncatemacro{\o}{\no}
	\pgfmathtruncatemacro{\no}{\o+2}
	\draw[orangepath,medge] (v-\no) to [bend left] (v-\o);
	\pgfmathtruncatemacro{\o}{\no}
	\pgfmathtruncatemacro{\no}{\o-1}
	\draw[orangepath,medge] (v-\no) to (v-\o);
	\pgfmathtruncatemacro{\o}{\no}
	\pgfmathtruncatemacro{\no}{\o+2}
	\draw[orangepath,medge] (v-\no) to [bend left] (v-\o);
	\pgfmathtruncatemacro{\o}{\no}
	\pgfmathtruncatemacro{\no}{\o+2}
	\draw[orangepath,medge] (v-\no) to [bend left] (v-\o);
	\pgfmathtruncatemacro{\o}{\no}
	\pgfmathtruncatemacro{\no}{\o+1}
	\draw[orangepath,medge8] (v-\no) to (v-\o);
	\draw[orangepath,medge] (t) to[ out=225,in=225,"$Q_1$" left] (v-\no);
	
	% purple
	\pgfmathtruncatemacro{\o}{\no}
	\pgfmathtruncatemacro{\no}{\o+3}
	\draw[purpleedges,medge5] (v-\no) to[ bend right] (v-\o);
	\draw[purpleedges,medge5] (v-\o) to[ bend right] (v-\no);
	\pgfmathtruncatemacro{\r}{\nr}
	\pgfmathtruncatemacro{\nr}{\r+1}
	\draw[purpleedges,medge7] (v-\nr) to[ bend right] (v-\r);
	\draw[purpleedges,medge7] (v-\r) to[ bend right] (v-\nr);
	
	\def\tikzsegment#1#2#3{ % This is the macro explained above
		\path let
		\p1=($(#3)-(#2)$),
		\n1={veclen(\p1)*0.7}
		in (#2) -- (#3) 
		node[minimum width=\n1, 
		inner sep=0pt, 
		pos=0.5,sloped,rectangle,rounded corners,
		#1] 
		(line){};
	}
	
	\newcommand{\drawCut}[2] {
		\tikzsegment{very thick,dash dot,draw,minimum height=0.31cm}{v-#1}{v-#2};
	}
	
	\foreach \i in {0,1} {
		\pgfmathtruncatemacro{\a}{13+\i * 16};
		\pgfmathtruncatemacro{\b}{16 + \i * 16};
		\drawCut{\a}{\b};
		\pgfmathtruncatemacro{\a}{25+\i * 16};
		\pgfmathtruncatemacro{\b}{20 + \i * 16};
		\drawCut{\a}{\b};
	}

	\draw (pp1) -- (p1);
\draw (pp2) -- (p2);	

\end{tikzpicture}
	
	\caption{\label{fig:diameter}Three types of gadgets used for the construction of $G_\ell$, and the graph $G_2$.
		The orange arcs and the red arcs in $G_2$ are the arcs of $Q_1$ and $Q_2$ respectively.
		$P_1$ and $P_2$ are the $(s,t)$-paths on the left and the right respectively. Blue arcs are the arcs that do not belong to neither of $P_1, P_2$ and $Q_1, Q_2$. Dashed rounded rectangles show the boundaries of the gadgets used in the construction.}
\end{figure}

\medskip
\noindent \textbf{Paths $P_1,P_2$ and $Q_1,Q_2$.}
The four paths we describe are clearly seen in Figure~\ref{fig:diameter}d, where the graph $G_2$ is presented.
By construction, there are two internally disjoint $(s,t)$-paths in $G_\ell$.
The path $P_1$ starts in $s$, then goes through vertices $s_1, s_2$ up to $s_8$ of the source gadget.
Then the vertices $h_2$ and $h_3$ of the first hat gadget follow, then the vertices $h_5, h_6$ up to $h_{10}$ of the second hat gadget follow on $P_1$, and so on.
The path $P_1$ ends with vertices $t_8, t_7$ down to $t_1$ and, finally, the vertex $t$.
The path $P_2$ starts in $s$, follows $s_9$ through $s_{14}$, and ends with $t_{14}$ through $t_9$ in $t$.
The paths $Q_1$ and $Q_2$ are two $(t,s)$-paths in $G_\ell$.
Their construction is clear from Figure~\ref{fig:diameter}d, where the arcs of each of them receive specific color.

Note that we used $2\ell+1$ gadgets to construct $G_\ell$.
Any of these gadgets is an induced subgraph of $G_\ell$.
Two gadgets can share either zero or two vertices, however, they cannot share any arc of $G_\ell$.
Moreover, each arc of $G_\ell$ belongs to exactly one gadget.
From now on, by \textit{gadgets} we refer to these $2\ell+1$ induced subgraphs of $G_\ell$.

\medskip \noindent \textbf{Separating and containing gadgets.}
We say that a gadget \emph{separates} two distinct vertices $u$ and $v$ in $G_\ell$ if there is no $(u,v)$-path and $(v,u)$-path in $G_\ell-X$, where $X$ is the arc set of the gadget.
We say that a gadget \emph{strictly contains} vertex $v \in V(G_\ell)$ if $v$ belongs to the vertex set of this gadget and does not belong to the vertex set of any other gadget in $G_\ell$.
Thus, there exist vertices that are not strictly contained in any gadget.
We observe some trivial facts about the gadgets.

\begin{observation}
	The following holds:
	\begin{enumerate}
		\item For every vertex $v \in V(G_\ell)$, if a gadget strictly contains $v$, then this gadget separates $v$ and every $u\in V(G_\ell)\setminus\{v\}$.
		\item For every distinct $u, v \in V(G_\ell)$, if $u$ and $v$ are not separated by any gadget in $G_\ell$, then $\{u,v\}$ is the intersection of the vertex sets of two gadgets in $G_\ell$.
		\item For every distinct $u, v \in V(G_\ell)$, if two gadgets both separate $u$ and $v$ in $G_\ell$ and share two vertices $u'$ and $v'$, then there is no path between $u$ and $v$ in $G_\ell-\{u',v'\}$.
	\end{enumerate}
\end{observation}

\begin{lemma}
	If a gadget separates $u$ and $v$ in $G_\ell$, and does not strictly contain neither $u$ nor $v$, then the arcs of any $(u,v)$-path in $G_\ell$ induce a single path inside this gadget.
\end{lemma}
\begin{proof}
	First, since the gadget is a separating gadget, it should contain at least one arc of any $(u,v)$-path in $G_\ell$.
	Hence, the arcs of the $(u,v)$-path induce one or several paths inside the gadget.
	The gadget is a hat gadget, since the source and the sink gadget can be separating only to the vertices they strictly contain.
	Any path induced in the gadget by the arcs of the $(u,v)$-path has endpoints in $\{h_1,h_3,h_4,h_{10}\}$, since both $u$ and $v$ are not strictly contained in the gadget.
	Then one of the paths induced in the gadget should be a path between $\{h_1,h_4\}$ and $\{h_3,h_{10}\}$, otherwise the gadget is not separating.
	Hence, if there are more than one paths induced, then there are exactly two paths induced, and the second path is also a path between $\{h_1,h_4\}$ and $\{h_3,h_{10}\}$.
	Since the gadget is separating and $\{h_1,h_4\}$ with $\{h_3,h_{10}\}$ are boundaries of the separation, we should have an odd number of paths going between them.
	Thus, there is only one induced path.
\end{proof}

\begin{lemma}
	${G_\ell}$ is 2-strongly-connected.
\end{lemma}
\begin{proof}
First note that ${G_\ell}$ is strongly-connected, since every vertex of ${G_\ell}$ is reachable from $s$ and $t$ is reachable from every vertex of ${G_\ell}$ and there is a $(t,s)$-path in ${G_\ell}$.
We now need to show that ${G_\ell}-v$ is strongly-connected for every vertex $v\in V({G_\ell})$.

Since ${G_\ell}$ is strongly-connected and ${G_\ell}[N^-_{G_\ell}(s)\cup N^+_{G_\ell}(s)]$ is strongly-connected, we have that ${G_\ell}-s$ is strongly-connected.
The same argument works for $t$, so ${G_\ell}-t$ is strongly-connected as well.
We shall now prove that ${G_\ell}-v$ is strongly-connected for an arbitrary vertex $v \in V({G_\ell})\setminus\{s,t\}$.
Since there are two disjoint $(s,t)$-paths in ${G_\ell}$, $t$ is reachable from $s$ in $({G_\ell},v)$.
Analogously, $s$ is reachable from $t$ in $({G_\ell},v)$. 
Hence, to prove that ${G_\ell}-v$ is strongly-connected it is enough to show for every vertex $u\in V({G_\ell}-v)\setminus\{s,t\}$ that there is an $(s,u)$-path or a $(t,u)$-path in ${G_\ell}$ and there is an $(u,s)$-path or a $(u,t)$-path in ${G_\ell}-v$.

Take a vertex $u \in V({G_\ell})\setminus\{v,s,t\}$.
First, note that $u $ lies on an $(s,t)$-path in ${G_\ell}$.
If $v$ does not belong to this path, then this path remains in ${G_\ell}-v$ yielding an $(s,u)$-path and an $(u,t)$-path in ${G_\ell}$ as required.
We now assume that $v$ belongs to every $(s,t)$-path in ${G_\ell}$ containing $u$.
Then exactly one of $(s,u)$-path and $(u,t)$-path exists in ${G_\ell}-v$.

\begin{claim}\label{claim:disjoint-paths}
	If $u \in V({G_\ell})\setminus \{s,t\}$ lies on $Q_1$ or $Q_2$ in ${G_\ell}$, then there exists
	\begin{itemize}
		\item an $(s,u)$-path and a $(t,u)$-path that are internally disjoint;
		\item an $(u,t)$-path and an $(u,s)$-path that are internally disjoint.
	\end{itemize}
\end{claim}
\begin{proof}[Proof of Claim~\ref{claim:disjoint-paths}]
	Let $i$ be such that $u$ belongs to $P_i$ and $j$ be such that $u$ belongs to $Q_i$.
	Note that the vertices in $V(Q_j)\cap V(P_i)$ appear on $Q_j$ in the order reverse to the order on $P_i$.
	Then the $(t,u)$-subpath of $Q_j$ only uses vertices of $P_i$ that appear after $u$ on $P_i$.
	Thus, the $(s,u)$-subpath of $P_i$ does not share any internal vertex with a $(t,u)$-subpath of $Q_j$.
	These two paths form the first pair of the claim.
	
	To proceed with the second pair of paths, take the $(u,s)$-subpath of $Q_j$ and the $(u,t)$-subpath of $P_i$.
	Since the $(u,s)$-subpath uses only those vertices of $P_i$ that appear before $u$ on $P_i$, these paths are also internally disjoint.
\end{proof}
Suppose that $u$ lies on a $(t,s)$-path in ${G_\ell}$.
From the claim follows that $u$ is reachable from $s$ or $t$ and $s$ or $t$ is reachable from $u$ in ${G_\ell}-v$ and we are done.
Hence, we can assume that $u$ does not lie neither on $Q_1$ nor $Q_2$.
Then $u$ is a vertex in $N^+_{G_\ell}(s)\cup N^-_{G_\ell}(t)$ or a vertex in the inner cycle of a hat gadget, i.e.\ one of the vertices among the vertices $h_2,h_6,h_7,h_8$ of some hat gadget.

If $u$ is a vertex in $N^+_{G_\ell}(s)$, then there are two internally disjoint paths from $s$ to $u$ and two internally disjoint paths from $u$ to $s$ in ${G_\ell}$, so both $(u,s)$-path and $(s,u)$-path exist in ${G_\ell}-v$ and we are done.
Analogously, if $u$ is a vertex in $N^-_{G_\ell}(t)$, we have both $(u,t)$-path and $(t,u)$-path in ${G_\ell}-v$.

It is left to consider the case when $u$ is a vertex of the inner cycle of a hat gadget.
Note that there exist two internally disjoint paths starting in distinct vertices of $V(Q_i)$ and ending in $u$, for each $i \in \{1,2\}$ in ${G_\ell}$.
Symmetrically, we have two internally disjoint paths starting in $u$ and ending in distinct vertices of $V(Q_i)$, for each $i \in \{1,2\}$ in ${G_\ell}$.
Since at least one of $Q_1$ and $Q_2$ is presented in ${G_\ell}-v$, we have a $(t,u)$-path and an $(u,s)$-path in ${G_\ell}$.
The proof is complete.
\end{proof}

\begin{lemma}
$\diam({G_\ell})=\dist_{G_\ell}(s,t)=8\ell+10$.
\end{lemma}
\begin{proof}
	Let $d$ be the length of $P_1$ and $P_2$ in ${G_\ell}$.
	First, note that $\dist_{G_\ell}(s,t)$ equals $d$.
	Indeed, $V({G_\ell})=V(P_1)\cup V(P_2)$, so an $(s,t)$-path in ${G_\ell}$ consists only of vertices of $P_1$ and $P_2$.
	Note that there is no arc in ${G_\ell}$ going from an $i^\text{th}$ vertex on $P_x$ to a $j^\text{th}$ vertex on $P_y$ for every choice of $x,y\in \{1,2\}$ and $i, j \in [d+1]$ with $i + 1 < j$.
	Hence, to reach $t$ from $s$ one should use at least $d$ arcs in ${G_\ell}$.
	Our goal is to prove that $\dist_{G_\ell}(u,v)\le d$ for each $u,v\in V({G_\ell})$.
	
	Now denote the internal vertices in $P_1$ by $a_1, a_2, \ldots, a_{d-1}$ in the order corresponding to $P_1$.
	Thus, $\dist_{G_\ell}(s,v)\le i$ and $\dist_{G_\ell}(v,t)\le d-i$ for each $i \in [d-1]$.
	Note that it holds that $\dist_{G_\ell}(s,v)=i$ and $\dist_{G_\ell}(v,t)=d-i$, since otherwise $\dist_{G_\ell}(s,t)\le \dist_{G_\ell}(s,v)+\dist_{G_\ell}(v,t)<i+(d-i)=d$.
	Analogously denote the internal vertices in $P_2$ by $b_1, b_2, \ldots, b_{d-1}$.
	It holds that $\dist_{G_\ell}(s,b_i)=i$ and $\dist_{G_\ell}(b_i,t)=d-i$ for each $i \in [d-1]$.
	
	Now consider the distance $\dist_{G_\ell}(t,s)$.
	We know that $\dist_{G_\ell}(t,s)$ is at most the length of $Q_1$ or $Q_2$.
	Observe that $Q_1$ uses four arcs in the sink gadget, three arcs in $\ell$ hat gadgets, one arc in $\ell-1$ hat gadgets and four arcs in the source gadget.
	Hence, the length of $Q_1$ is $4+3\ell+(\ell-1)+4=4\ell+7$.
	As for $Q_2$, it uses six arcs in each of the sink and the source gadgets, one arc in $\ell$ hat gadgets and three arcs in $\ell-1$ hat gadgets.
	The length of $Q_2$ is $12+\ell+3(\ell-1)=4\ell+9$.
	Hence, $\dist_{G_\ell}(t,s)\le 4\ell+7<d$.
	
	Now take a vertex $v\in \{a_1,b_1,a_2,b_2,\ldots, a_{d-1},b_{d-1}\}$ and consider the distance $\dist_{G_\ell}(t,v)$.
	If $v$ belongs to $Q_1$ or $Q_2$, then $\dist_{G_\ell}(t,v)<4\ell+9<d$.
	If $v$ does not belong to $Q_1$ and $Q_2$, then $v \in N^+_{G_\ell}(s)$ or $v \in N^-_{G_\ell}(t)$ or $v$ is a vertex of the inner cycle of a hat gadget.
	In the first case, $\dist_{G_\ell}(t,v)\le \dist_{G_\ell}(t,s)+1\le 4\ell+8<d$.
	In the case $v \in N^-_{G_\ell}(t)$, $\dist_{G_\ell}(t,v)=2$.
	For $v$ being a vertex of some inner cycle, note that $v$ is reachable from a vertex of $Q_1$ or a vertex of $Q_2$ using a single arc.
	Hence, $\dist_{G_\ell}(t,v)< 4\ell+9+1<d$.
	
	To handle the distance $\dist_{G_\ell}(v,s)$, we note that the graph ${G_\ell}$ is isomorphic to the graph ${G_\ell}^T$ with isomorphism $f: V({G_\ell})\to V({G_\ell})$, such that $f(s)=t$, $f(t)=s$, $f(a_i)=a_{d-i}$ and $f(b_i)=b_{d-i}$ for each $i \in [d-1]$.
	Thus, $\dist_{G_\ell}(v,s)=\dist_{{G_\ell}^T}(f(v),t)=\dist_{G_\ell}(t,f(v))$.
	Since $f(v)$ is $a_i$ or $b_i$ for some $i \in [d-1]$, we have that $\dist_{G_\ell}(v,s)\le 4\ell+9<d$.

	It is left to prove that $\dist_{G_\ell}(u,v)\le d$ for every choice of $u,v\in \{a_1,b_1,a_2,b_2,\ldots, a_{d-1},b_{d-1}\}$.
	It is easy to see that for every $1\le i\le j <d$ $\dist_{G_\ell}(v,a_j)=j-i$ and $\dist_{G_\ell}(b_i,b_j)=j-i$, otherwise $\dist_{G_\ell}(s,t)<d$.
	
	Note that for each $i \in [d-1]$ with $i$ giving $1, 2$ or $3$ modulo $8$, we have an arc from $a_i$ to $b_j$ in ${G_\ell}$, where $j$ is such that $|j-i|=1$.
	Then, for each $i \in [d-1]$ with $i$ giving $6,7$ or $0$ modulo $8$ and $i\in\{1,2,3,d-1,d-2,d-3\}$, we have an arc from $b_i$ to $a_j$ in ${G_\ell}$, where $j$ is such that $|j-i|=1$.
	It follows that for each $i \in [d-1]$, we have an arc from $a_{j}$ to $b_t$, where $i \le j \le i + 5$ and $|t-j|=1$.
	Hence, $\dist_{G_\ell}(a_i,b_{j+1})\le \dist_{G_\ell}(a_i,a_j)+1+\max\{\dist_{G_\ell}(b_{j-1},b_{j+1}),\dist_{G_\ell}(b_{j+1},b_{j+1})\}\le j-i+1+2=j-i+3.$
	Then $\dist_{G_\ell}(a_i,b_{i+6})\le \dist_{G_\ell}(a_i,b_{j+1})+\dist_{G_\ell}(b_{j+1},b_{i+6})\le (j-i+3)+(i-j+5)=8$.
	Analogously, we have that $\dist_{G_\ell}(b_i,a_{i+6})\le 8$.
	We conclude that $\dist_{G_\ell}(a_i,b_j)\le j-i+2\le d$ and $\dist_{G_\ell}(b_i,a_j)\le j-i+2\le d$ for each $i, j \in [d-1]$ such that $j-i\ge 6$.
	
	It remains to consider distances of form $\dist_{G_\ell}(u,v)$, where $u\in\{a_i,b_i\}$ and $v\in \{a_j,b_j\}$, such that $j-i\le 5$.
	Then we have that both $u$ and $v$ are vertices of the same gadget or two adjacent gadgets.
	Diameter of any gadget is at most eight, so clearly the distance between $u$ and $v$ is at most $16$, which is less than $8\ell+10$ for $\ell\ge 1$.
\end{proof}

\begin{lemma}
	Let $u, v \in V({G_\ell})$ be a pair of vertices in ${G_\ell}$ and let $T$ be an $(u,v)$-path in ${G_\ell}$.
	If $T$ contains an arc of some gadget (either source, sink or hat gadget) of ${G_\ell}$, then
	\begin{itemize}
		\item this gadget contains $u$ or $v$;
		\item this gadget separates $u$ and $v$ in ${G_\ell}$;
		\item $T$ contains exactly two arcs of this gadget.
		There are at most two such gadgets for $T$ overall.
	\end{itemize}
	\label{lemma:separation}
\end{lemma}
\begin{proof}
	We first consider the hat gadgets.
	Targeting towards a contradiction, suppose that there exists an $(u,v)$-path $T$ and a hat gadget that do not satisfy the lemma statement.
	Denote the vertices of the hat gadget in ${G_\ell}$ by $h_1, \ldots, h_{10}$, respectively to the definition of a hat gadget.
	The hat gadget does not separate $u$ and $v$ in $G_\ell$ and $u, v \notin \{h_1,h_2,\ldots,h_{10}\}$.
	Finally, $T$ uses more than two arcs of this gadget.
	
	Note that since the gadget does not separate $u$ and $v$ in ${G_\ell}$, the arcs of $T$ in the gadget either form a path between $h_1$ and $h_4$, or a path between $h_3$ and $h_{10}$, or two disjoint paths, one going from $\{h_1, h_4\}$ to $\{h_3, h_{10}\}$, and one going in the opposite direction from $\{h_3,h_{10}\}$ to $\{h_1,h_4\}$.
	Other cases are not possible since $\{h_1,h_4\}$ and $\{h_3,h_{10}\}$ both are cuts of ${G_\ell}$.
	
	Consider first that this is a path between $h_1$ and $h_4$.
	First note that the only $(h_4,h_1)$-path inside the gadget uses exactly two arcs, $(h_4,h_5)$ and $(h_5,h_1)$, but $T$ uses more than two arcs.
	If this is a $(h_1,h_4)$-path, then $T$ contains a vertex $x$ outside the gadget such that $(x,h_1)\in A({G_\ell})$.
	Symmetrically, $T$ contains a vertex $y\notin\{h_1,\ldots, h_{10}\}$ such that $(h_4,y)\in A({G_\ell})$, and $x$ and $y$ are distinct.
	However, the indegree and outdegree of both $h_1$ and $h_4$ are equal to two, so there is only one option for $x$ and only one option for $y$ in ${G_\ell}$.
	However, from the construction of ${G_\ell}$ it is clear that both $x$ and $y$ should be a predecessor of $h_1$ on either $P_1$ or $P_2$ in ${G_\ell}$.
	Hence, the case of a path between $h_1$ and $h_4$ is not possible.
	
	The case of a path between $h_3$ and $h_{10}$ is symmetrical.
	The only $(h_3, h_{10})$-path consists of only two arcs, and the $(h_{10},h_3)$-path is not possible since $T$ cannot contain a $(h_{10}, h_3)$-subpath, as the only outside ingoing neighbour of $h_{10}$ is the only outside outgoing neighbour of $h_3$.
	
	The only case left for the intersection of $T$ and the arcs of the gadget is two disjoint paths.
	One path should start either in $h_1$ or $h_4$ and end either in $h_3$ or $h_{10}$, and the other should go in the opposite direction.
	Suppose that one of the paths starts in $h_3$.
	Then it should use the only outgoing arc $(h_3, h_9)$.
	However, $\{h_3,h_{10}\}$ separates all paths going from $\{h_1,h_4\}$ to $h_{10}$, so a disjoint path in the other direction is not possible in this case.
	Hence, one of the paths should start in $h_{10}$.
	The only arc outgoing of $h_{10}$ is $(h_{10},h_4)$, so one of the paths is a $(h_{10},h_4)$-path and the other is an $(h_1,h_3)$-path inside the gadget.
	But then $T$ should contain an $(h_1,h_4)$-subpath in ${G_\ell}$.
	We already know that this is not possible.
	
	It is left to consider the case when the source gadget or the sink gadget does not contain $u$ and $v$, but $T$ has at least three arcs of this gadget.
	We consider the case of the source gadget, and the case of the sink gadget is symmetrical.
	Since $T$ does not start nor end in the source gadget, the arcs of $T$ should induce a path between $s_8$ and $s_{14}$ inside it.
	The only $(s_{14},s_8)$-path consists of just two arcs, so it should be an $(s_8,s_{14})$-path inside the gadget.
	Then $T$ should contain an arc $(x,s_8)$ and an arc $(s_{14},y)$, where $x,y$ are outside the gadget.
	The only choice for both $x$ and $y$ is the vertex $h_5$ of the first hat gadget in $G_\ell$.
	This is a contradiction since $x$ and $y$ should be distinct.
	
	For the very last sentence in the lemma statement, suppose that there are three gadgets containing an arc of $T$ but not containing $u$ and $v$ and not separating $u$ and $v$.
	Then one of the three gadgets separates the other two from each other, and this gadget is necessarily a hat gadget.
	Then the arcs of $T$ should form at least one path between $\{h_1,h_4\}$ and $\{h_3,h_{10}\}$ inside this separating gadget, but this is not the case.

\end{proof}

\begin{lemma}
	For $\ell\ge 17$, the longest path length in ${G_\ell}$ is $\dist_{G_\ell}(s,t)+4 = 8\ell + 14$. Additionally,
	for a vertex $v \in G_\ell$ that is either $h_6$ or $h_8$ in the median hat gadget (i.e. the $\ell$-th hat gadget out of $2\ell - 1$), the longest path in $G_\ell$ that ends in $v$ has length at most $4\ell + 15$, and exactly $4\ell + 15$ for $h_8$.
	\label{lemma:Gk_paths}
\end{lemma}
\begin{proof}
	Consider two vertices  $u$ and $v$ and an $(u,v)$-path $T$. We show that the length of $T$ is at most $d+4$, where $d=\dist_{G_\ell}(s,t)=8\ell+10$.
	For any hat gadget not containing $u$ and $v$ but separating $u$ and $v$, we know that the arcs of $T$ form a path between $\{h_1,h_4\}$ and $\{h_3, h_{10}\}$ inside this gadget.
	
	We first consider the case when $\dist_{G_\ell}(s,u)>\dist_{G_\ell}(s,v)$.
	Then in any separating hat gadget that does not contain $u$ or $v$ strictly the arcs of $T$ form a path from $\{h_3,h_{10}\}$ to $\{h_1, h_4\}$.
	Denote all separating gadgets by $H_1, H_2, \ldots, H_p$ in the order in which $T$ traverses them.
	Denote by $x_i$ the number of arcs of $T$ inside $H_i$ for each $i \in [p]$.
	Note that each $x_i \in [5]$, as the longest path from $\{h_3, h_{10}\}$ to $\{h_1, h_4\}$ in $H_i$ is an $(h_3,h_1)$-path of length five.
	However, if the path induced by $T$ in $H_i$ ends in $h_1$, then the path induced by $T$ in $H_{i+1}$ starts in $h_{10}$.
	A longest path starting in $h_{10}$ has length three if it ends in $h_1$, and length one if it ends in $h_4$.
	For each $i \in [p]$, let $c_i \in \{0,1\}$ be equal to $1$ if the path induced by $T$ in $H_i$ ends in $h_1$, and to $0$ otherwise.
	Additionally, let $c_0 = 1$ if the path induced by $T$ in $H_1$ starts in $h_3$, otherwise $c_0=0$.
	
	For each $i \in \{0,\ldots,p-1\}$, consider the values of $c_i$ and $c_{i+1}$.
	If $c_i=0$, then the path induced by $T$ in $H_{i+1}$ starts in $h_3$, otherwise it starts in $h_{10}$.
	If $c_{i+1}=0$, then the path induced by $T$ in $H_{i+1}$ ends in $h_4$, otherwise it ends in $h_1$.
	If $(c_i,c_{i+1})=(0,0)$, then the path inside $H_{i+1}$ is an $(h_3,h_4)$-path and has length three.
	If $(c_i,c_{i+1})=(0,1)$, then this path goes from $h_{3}$ to $h_1$ and has length five.
	If $(c_i,c_{i+1})=(1,0)$, then the path is an $(h_{10},h_4)$-path and has length one.
	Finally, if the pair equals $(1,1)$, then inside $H_{i+1}$ we have an $(h_{10},h_1)$-path of length three.
	Thus, the formula for $x_{i+1}$ is $3-2c_i+2c_{i-1}$.
	We obtain that $\sum_{i=1}^p x_i=3p+2c_0-2c_p\le 3p+2$.

	Now observe that apart from arcs inside separating gadgets, $T$ can contain arcs inside gadgets containing $u$ and $v$ but not separating them plus four arcs in two additional gadgets.
	Since each gadget consists of at most $14$ vertices, $T$ can have at most $13$ arcs inside gadgets containing $u$ and $v$.
	The vertex $u$ is either strictly contained in a gadget, or is non-strictly contained in two adjacent gadgets. In the second case, one of the two gadgets is necessarily
	a separating hat gadget. 
	Since the same holds $v$, there are at most two gadgets in $G_\ell$ that contain $u$ or $v$ but not separate them.
	Hence, the length of $T$ is at most $(3p+2)+2\cdot 13+4\le 3p + 32\le 3(2\ell-1)+32\le 6\ell+29$.
	For $\ell \ge 17$, we have that $6\ell+29 < 8\ell + 14$, which is the desired bound on the length of $T$.
	For the second part of the statement, if $v$ is strictly contained the median hat gadget, then $p$ is at most $\ell - 1$, and the length of $T$ is at most $3\ell + 32 \le 4\ell + 15$.
	
	We move on to the case $\dist_{G_\ell}(s,u)\le \dist_{G_\ell}(s,v)$.
	If $u$ and $v$ are in the same gadget or are in two adjacent gadgets, then the length of $T$ is at most $4\cdot 14+4\le 60<d+4$.
	Hence, there is at least one hat gadget separating $u$ and $v$ but not containing $u$ or $v$.
	Consider the subpath of $T$ that is formed by arcs of the separating gadgets that do not strictly contain $u$ or $v$.
	Let the starting point of this path be $u'$ and ending point be $v'$.
	Note that in any separating hat gadget, $T$ induces a path from $h_i \in \{h_1,h_4\}$ to $h_j \in \{h_3,h_{10}\}$.
	Note that the length of such path always equals $\dist_{G_\ell}(s,h_j)-\dist_{G_\ell}(s,h_i)$.
	Hence, the length of the $(u',v')$-subpath of $T$ is exactly $\dist_{G_\ell}(s,v')-\dist_{G_\ell}(s,u')$.
	It remains to estimate the length of the $(u,u')$-subpath and the $(v',v)$-subpath of $T$, consider the following expression for the length of the $(u,v)$-path:
    \begin{multline*}
        \length(T) = \length \left(T_{(u, u')}\right) + \dist_{G_\ell}(s,v')-\dist_{G_\ell}(s,u') + \length \left(T_{(v', v)}\right) =\\ d + \left(\length \left(T_{(u, u')}\right)-\dist_{G_\ell}(s,u')\right) + \left(\length \left(T_{(v', v)}\right) - \dist_{G_\ell}(v', t)\right),
    \end{multline*}
    where we denote the $(u, u')$-subpath of $T$ by $T_{(u, u')}$, and the $(v', v)$-subpath of $T$  by $T_{(v', v)}$.

    It suffices to show that the length of the $(u, u')$-subpath is at most $\dist_{G_\ell}(s, u') + 2$, then, by symmetry, each of the last two terms above is at most two. There are two cases, either $u$ belongs to the source gadget, or $u$ belongs only to hat gadgets. We start with the first case, $u'$ is then either $s_8$ or $s_{14}$ in the source gadget. The following claim completes this case.

    \begin{claim}\label{claim:source}
        In the source gadget, for $w \in \{s_8, s_{14}\}$, the longest path that ends in $w$ has length at most $\dist(s, w) + 2$.
    \end{claim}
    \begin{proof}[Proof of Claim~\ref{claim:source}]
        Let us call the arcs that inrease (resp. decrease) the distance to $s$ forward (resp. backward) arcs. 
        Observe that if a path that ends in $w$ does not take backward arcs, its length is at most $\dist(s, w)$.
        We now consider the choice of the first backward arc of the potential path. To recall the vertex numeration in the source gadget, see Figure~\ref{fig:diameter}a.

        \begin{itemize}[leftmargin=1.5cm]
            \item[$\boxed{s_2s}$] After $s$, the path has to proceed to $s_{10}$, since $\{s_2, s_{10}\}$ is a cut that separates $s$  from $\{s_8, s_{14}\}$.
                The longest choice for such a subpath starts in $s_1$ or $s_9$ and collects all vertices on one side of the cut, e.g. $s_1s_2ss_9s_{10}$.
                Then the path has the only option to proceed to $s_{14}$. if $w = s_{14}$ this case is settled as the length of the path is at most $8 = \dist(s, s_{14}) + 2$. If $w = s_8$, then the path has to take the arc $s_{14}s_7$, and then the arc $s_7s_8$; the arc $s_7s_5$ cannot be taken since $s_7$ separates $s_8$ from $s_5$. Again, the length of the path is at most $10 = \dist(s, s_8) + 2$.
            \item[$\boxed{s_{10}s}$] Identically to the previous case, the first part of the path ends in $s_2$ and takes at most 4 arcs inside the vertex set $\{s, s_1, s_2, s_9, s_{10}\}$.
                Then the path has to proceed along $P_1$ until it either takes the arc $s_4s_{13}$, and then the argument is exactly the same as in the previous case, or proceeds straight to $s_8$, thus in the case $w = s_8$ yielding  a path of length at most $10$.
            \item[$\boxed{s_2s_9}$] Exactly as in the case $\boxed{s_2s}$, the path has to reach $s_{10}$, and then the analysis is the same.
            \item[$\boxed{s_{10}s_1}$] Symmetrically to the previous case, the analysis repeats the case $\boxed{s_{10}s}$.
            \item[$\boxed{s_3s_{10}}$] After $s_{10}$, the path has to proceed to $s_{11}$ as $\{s_3, s_{10}\}$ is a separator, the analysis regarding the subpath after $s_{10}$ is identical to the case $\boxed{s_2s}$.
                Observe that the subpath leading to $s_3$ cannot take both vertices $s_1$ and $s_9$, since $s_2$ together with $s_{10}$ separates $s_1$ and $s_9$ from $s_3$. Therefore, prior to reaching $s_{10}$ the path takes at most 4 arcs, leading again to the total length of at most $\dist(s, w) + 2$.
            \item[$\boxed{s_{11}s_2}$] Symmetrically to the previous case, the analysis reduces to the case $\boxed{s_{10}s}$.
            \item[$\boxed{s_6s_3}$] Before $s_6$, only the vertices $s_4$ and $s_5$ can be taken since $s_6s_3$ is the first backward arc. After $s_3$, the path either proceeds to $s_{10}$, or $s_4$ and then $s_{13}$. In the first case, the path has to go directly from $s_{10}$ to $w$ via forward arcs, yielding a length of at most 8 or 10 for $w = s_{14}$ and $w = s_8$ respectively. In the second case, no other backward arc can be taken as well, and the resulting path can only be shorter.
            \item[$\boxed{s_5s_4}$] Since $s_5s_4$ is the first backward arc on the path, the path has to start at $s_5$. Afterwards, the arc $s_4s_{13}$ has to be taken, and then only forward arcs to $s_{14}$, resulting in the path of length 3, or $s_8$, resulting in length 5.
            \item[$\boxed{s_7s_5}$] The target endpoint $w$ cannot be $s_8$, as $s_7$ separates $s_8$ from $s_5$.
                If $w = s_{14}$, the path has to either continue through $s_5s_4$ and $s_4s_{13}$, or through $s_5s_6$ and $s_6s_3$.
                In the first case, the path has to finish immediately by taking the arc $s_{13}s_{14}$, and before $s_7$ only the vertex $s_6$ can be taken, resulting in the length of at most 5.
                In the second case, the path has to start at $s_7$ as $s_6$ is taken and $s_{14}$ is the end of the path, and then similarly to  the case $\boxed{s_6s_3}$ the only option is to proceed from $s_3$ to $s_{10}$ and then to $s_{14}$ along forward arcs, resulting in length 8, or from $s_3$ directly to $s_{14}$ via forward arcs, resulting in length 6.
            \item[$\boxed{s_8s_6}$] The endpoint $w$ must be $s_{14}$. After $s_6$, the path either goes through $s_6s_3$, or through $s_6s_7s_5s_4s_{13}s_{14}$. In the latter case, $s_8$ has to be the starting vertex of the path and the legnth is 6. In the former, the starting vertex is either $s_7$ or $s_8$, and then the analysis is identical to the case where $s_6s_3$ is the first backward arc.
            \item[$\boxed{s_{12}s_{11}}$] The vertex $s_{12}$ has to be the starting point of the path. Afterwards, the arc $s_{11}s_2$ has to be taken, and then only forward arcs. The length of the path is then exactly $\dist(s, w)$.
            \item[$\boxed{s_{13}s_{12}}$] The path has to start at $s_{13}$, as $\{s_4s_{13}\}$ separates $s_{12}$ from $w$. The only possible continuation is $s_{12}s_{11}s_2s_3s_4s_5s_6s_7s_8$, resulting in the path of length 9 ending at $s_8$.
        \end{itemize}
        As the above cases cover all backward arcs in the source gadget, the proof of the claim is concluded.
    \end{proof}

    In the other case, $u$ does not belong to the source gadget, so the first hat gadget does not separate $u$ and $v$.
    Consider the topmost hat gadget that contains $u$. The vertex $u'$ has to be either $h_3$ or $h_{10}$ in this hat gadget, since the next hat gadget necessarily separates $u$ and $v$.
    Thus, $\dist_{G_\ell}(s, u') \ge 10$. Now it suffices to show that the length of the $(u, u')$-subpath is at most twelve. From Lemma~\ref{lemma:separation}, it follows that this subpath lies completely inside the gadget, except for the two arcs of the gadget above that form a path from $h_4$ to $h_1$ through a vertex outside of the gadget. Therefore, the subpath visits at most eleven vertices, the number of vertices in a hat gadget plus one extra vertex, and thus cannot exceed the length of twelve.
    To complete the proof, we show that there is indeed a path of length $8\ell + 14$ in $G_\ell$. The path proceeds as follows: start at the vertex $s_9$ of the source gadget, then go to $s_{10}$ and $s$, then proceed along the path $P_1$ of length $d$ until $t$ is reached, and finally go to $t_{10}$ and $t_9$.

    For the second part of the statement, consider $v \in \{h_6, h_8\}$ in the median hat gadget. Then $v'$ is either $h_4$ or $h_1$ in the same gadget.
    From the above, the length of the $(u, v')$-subpath of $T$ is at most $\dist_{G_\ell}(s, v') + 2$. Going from $v'$ to $v$, $T$ may have arcs inside the median hat gadget, and possibly two arcs in the next gadget that form a two-path between $h_3$ and $h_{10}$.
    Thus, the length of $T$ is at most $\dist_{G_\ell}(s, v') + 12$, since there are at most 10 available vertices to $T$ after $v'$. If $v'$ is $h_4$ in the gadget, then length of $T$ is at most $4\ell + 14$, since $\dist_{G_\ell}(s, v')$ is $4\ell + 2$ in this case. It only remains to consider the case
    where $v'$ is $h_1$ in the median hat gadget, $\dist_{G_\ell}(s, v') = 4\ell + 4$. It suffices to show that $T$ cannot take simultaneously all the vertices in the hat gadget, and the two arcs of the next gadget.
    Assume that happens, then $T$ has to go from $h_1$ to $h_3$ and then to $h_{10}$ through the two additional arcs. This leaves only the option to proceed  to $h_4$ and then along the remaining path. Thus, $h_9$ cannot be reached before $h_6$ or $h_8$, if the path is to collect all vertices.
    Finally, to see that there is a path of length $4\ell + 15$ ending at $h_8$, consider the following path. Start at $s_1$, go to $s_2$, then $s$, then proceed along $P_2$ until the vertex $h_1$ of the median hat gadget is reached. From there, complete the path with the following sequence:
    \[h_1h_2h_3wh_{10}h_4h_5h_6h_7h_8,\]
    where all the vertices are labelled corresponding to the hat gadget, and $w$ is the available vertex of the next gadget.
    This completes the lemma.
\end{proof}

Now we are ready to prove the hardness result of Theorem~\ref{thm:dichotomy_diameter}.

\begin{lemma}
    \probLPDiam on 2-strongly-connected directed graphs is \classNP-complete for $k \ge 5$.
    \label{lemma:diameter_np_complete}
\end{lemma}
\begin{proof}
    For each $k \ge 5$, we present a reduction from  \textsc{Hamiltonian Path} on undirected 2-connected graphs, for an intuitive illustration see Figure~\ref{fig:diameter_reduction}.
    Take a Hamiltonian path instance $H$ where $|V(H)| = n'$,
    we treat $H$ as a directed graph where every undirected edge is replaced by two directed arcs going in opposite directions.
    Assume that $n'$ has form $4\ell + (k - 5)$ for a certain integer $\ell \ge k/4 + 17$.
    For each vertex $w \in V(H)$, we construct the following instance of \probLPDiam. Consider a graph $C$ that we call a \emph{connector gadget}.
    The graph $C$ has four vertices $c_1$, $c_2$, $c_3$, $c_4$, and four arcs $c_3c_1$, $c_4c_2$, $c_1c_4$, $c_2c_3$. The resulting instance of \probLPDiam is the graph $G$ constructed by taking disjoint instances of the graphs $H$, $C$, $G_\ell$,
    and then associating the vertex $c_2$ of $C$ with $w$ in $H$, $c_1$ of $C$ with an arbitrary other vertex of $H$, $c_3$ of $C$ with the $h_6$ vertex of the median hat gadget of $G_\ell$, and $c_4$ with the $h_8$ vertex of the same gadget.
    This finishes the construction, and now we show the correctness of the reduction, that is,
    there is a Hamiltonian path in $H$ if and only if at least one of the constructed graphs $G$ has a path of length $\diam(G) + k$. 
    In the following claims, we show that $G$ shares most of the properties proved for $G_\ell$ above.

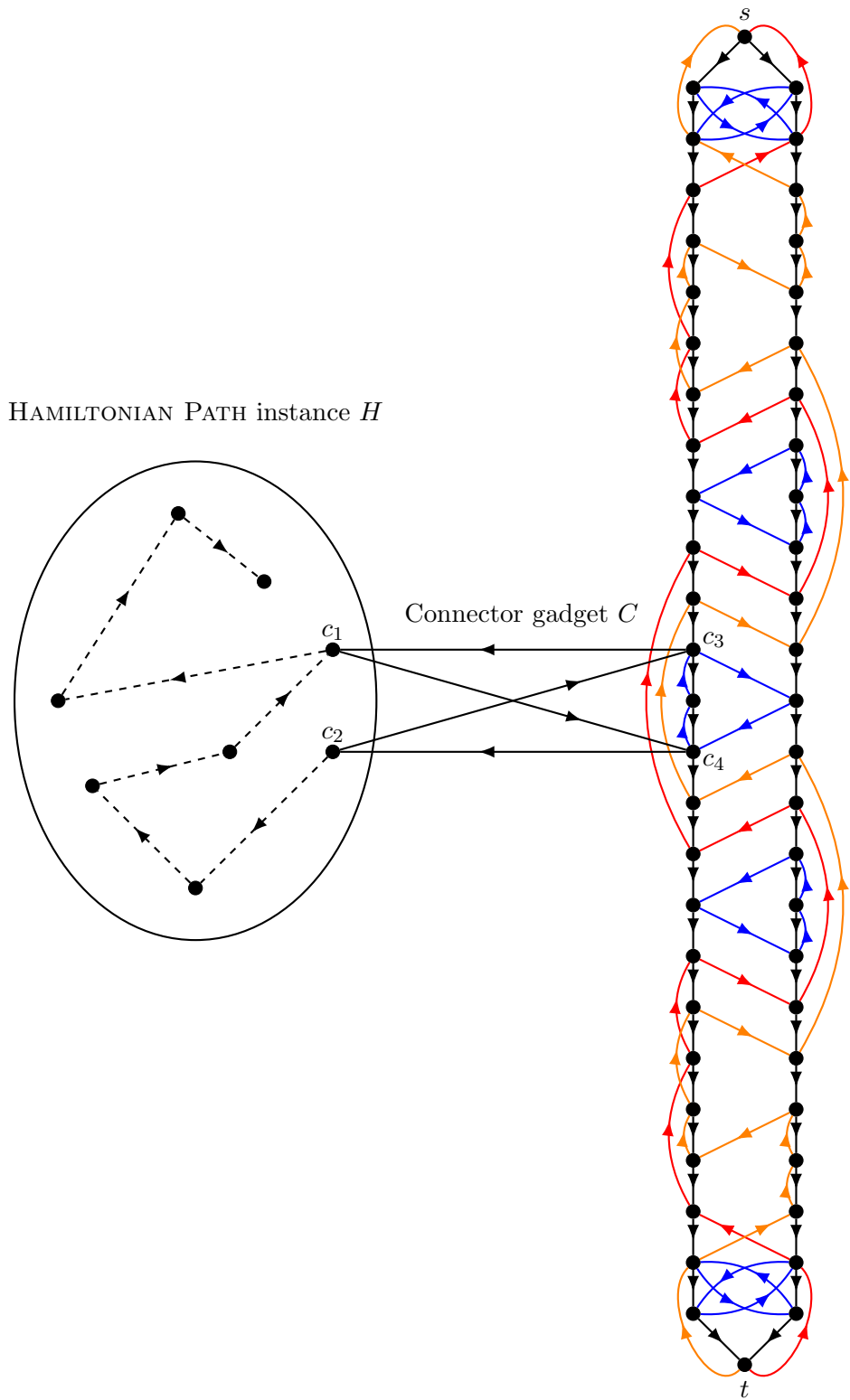
\begin{figure}	
	\centering
	\begin{tikzpicture}[scale=0.5,x={(0cm,-1cm)},y={(-1cm,0cm)}]
	\tikzstyle{vertex}=[circle, fill=black,draw,inner sep=0,minimum size=0.2cm]
	\usetikzlibrary{decorations.markings}
	\usetikzlibrary{arrows.meta}
	\tikzstyle{medge}=[thick,decoration={
		markings,
		mark=at position 0.6 with {\arrow{Latex}}},postaction={decorate}]
	\tikzstyle{medge5}=[thick,decoration={
		markings,
		mark=at position 0.5 with {\arrow{Latex}}},postaction={decorate}]
	\tikzstyle{medge7}=[thick,decoration={
		markings,
		mark=at position 0.7 with {\arrow{Latex}}},postaction={decorate}]
	\tikzstyle{medge8}=[thick,decoration={
		markings,
		mark=at position 0.8 with {\arrow{Latex}}},postaction={decorate}]
	\tikzstyle{dedge}=[thick,dashed, decoration={
		markings,
		mark=at position 0.6 with {\arrow{Latex}}},postaction={decorate}]
	\tikzstyle{redpath}=[red]
	\tikzstyle{orangepath}=[orange]
	\tikzstyle{purpleedges}=[blue]
	\usetikzlibrary{math} 
	\usetikzlibrary{patterns}
	
	\node [vertex,label=above:{$s$}] (s) at (-4,0) {};
	
	\def\n{25};
	\def\nm{24};
	\def\gap{1.5}

	\foreach \x in {1,...,\n} {
		\pgfmathtruncatemacro{\label}{\x*2}
		\node [vertex] (v-\label) at (-4+\gap*\x,1.5) {};
		\pgfmathtruncatemacro{\label}{\x*2+1}
		\node [vertex] (v-\label) at (-4+\gap*\x,-1.5) {};
	}
	
	\draw[medge] (s) to (v-2);
	\draw[medge] (s) to (v-3);
%	\node at ($(s)+(0,4)$) {e)};
	
	\foreach \x in {1,...,\nm} {
		\pgfmathtruncatemacro{\labela}{\x*2}
		\pgfmathtruncatemacro{\labelb}{\x*2+2}
		\draw[medge] (v-\labela) -- (v-\labelb);
		
		\pgfmathtruncatemacro{\labela}{\x*2+1}
		\pgfmathtruncatemacro{\labelb}{\x*2+3}
		\draw[medge] (v-\labela) -- (v-\labelb);
	}

	% red path here
	\draw[redpath,medge] (v-16) to [bend left] (v-12);
	\draw[redpath,medge] (v-12) to [bend left] (v-6);
	\draw[redpath,medge8] (v-6) to (v-5);
	\draw[redpath,medge] (v-5) to[out=45,in=45] (s);
	
	% orange
	\draw[orangepath,medge] (v-14) to [bend left] (v-10);
	\draw[orangepath,medge] (v-10) to [bend left] (v-8);
	\draw[orangepath,medge] (v-8) to (v-11);
	\draw[orangepath,medge] (v-11) to[ bend right] (v-9);
	\draw[orangepath,medge] (v-9) to[ bend right] (v-7);
	\draw[orangepath,medge8] (v-7) to (v-4);
	\draw[orangepath,medge] (v-4) to[out=135,in=135] (s);
	
	% purple
	\draw[purpleedges,medge7] (v-3) to[ bend right] (v-4);
	\draw[purpleedges,medge7] (v-4) to[ bend right] (v-3);
	\draw[purpleedges,medge5] (v-2) to[ bend right] (v-5);
	\draw[purpleedges,medge5] (v-5) to[ bend right] (v-2);
	
	% pattern
	\newcommand{\drawPatternUB}[2]
	{
		\pgfmathtruncatemacro{\r}{#1};
		\pgfmathtruncatemacro{\o}{#2};
		
		\pgfmathtruncatemacro{\nr}{\r-1};
		\pgfmathtruncatemacro{\no}{\o-1};
		
		\draw[redpath,medge] (v-\nr) to (v-\r);
		\draw[orangepath,medge] (v-\no) to (v-\o);
		
		\pgfmathtruncatemacro{\r}{\nr};
		\pgfmathtruncatemacro{\o}{\no};
		
		\pgfmathtruncatemacro{\nr}{\r + 8};
		\pgfmathtruncatemacro{\no}{\o + 12};
		
		\draw[redpath,medge] (v-\nr) to[ bend right] (v-\r);
		\draw[orangepath,medge] (v-\no) to[ bend right] (v-\o);
		
		\pgfmathtruncatemacro{\r}{\nr};
		\pgfmathtruncatemacro{\o}{\no};
		
		\pgfmathtruncatemacro{\a}{\r - 2};
		\pgfmathtruncatemacro{\b}{\a - 2};
		\pgfmathtruncatemacro{\c}{\b - 2};
		\pgfmathtruncatemacro{\d}{\c + 1};
		
		\draw[purpleedges,medge] (v-\a) to [ bend right] (v-\b);
		\draw[purpleedges,medge] (v-\b) to [ bend right] (v-\c);
		\draw[purpleedges,medge] (v-\c) to (v-\d);
		\draw[purpleedges,medge] (v-\d) to (v-\a);
	};
	
	\newcommand{\drawPatternBU}[2]
	{
		\pgfmathtruncatemacro{\r}{#1};
		\pgfmathtruncatemacro{\o}{#2};
		
		\pgfmathtruncatemacro{\nr}{\r-3};
		\pgfmathtruncatemacro{\no}{\o-3};
		
		\draw[redpath,medge] (v-\nr) to (v-\r);
		\draw[orangepath,medge] (v-\no) to (v-\o);
		
		\pgfmathtruncatemacro{\r}{\nr};
		\pgfmathtruncatemacro{\o}{\no};
		
		\pgfmathtruncatemacro{\nr}{\r + 12};
		\pgfmathtruncatemacro{\no}{\o + 8};
		
		\draw[redpath,medge] (v-\nr) to[bend left] (v-\r);
		\draw[orangepath,medge] (v-\no) to[bend left] (v-\o);
		
		\pgfmathtruncatemacro{\r}{\nr};
		\pgfmathtruncatemacro{\o}{\no};
		
		\pgfmathtruncatemacro{\a}{\o - 2};
		\pgfmathtruncatemacro{\b}{\a - 2};
		\pgfmathtruncatemacro{\c}{\b - 2};
		\pgfmathtruncatemacro{\d}{\c + 3};
		
		\draw[purpleedges,medge] (v-\a) to [bend left] (v-\b);
		\draw[purpleedges,medge] (v-\b) to [bend left] (v-\c);
		\draw[purpleedges,medge] (v-\c) to (v-\d);
		\draw[purpleedges,medge] (v-\d) to (v-\a);
	};

	\drawPatternUB{16}{14}
	
	\drawPatternBU{23}{25}
	
	\drawPatternUB{32}{30}
	
	\node [vertex,label=below:{$t$}] (t) at (-4+\gap*\n+\gap,0) {};
	\pgfmathtruncatemacro{\label}{\n*2}
	\draw[medge] (v-\label) to (t);
	\pgfmathtruncatemacro{\label}{\n*2+1};
	\draw[medge] (v-\label) to (t);
	
	\pgfmathtruncatemacro{\r}{(\n-7)*2}
	\pgfmathtruncatemacro{\nr}{\r+3}
	\draw[redpath,medge] (v-\r) to (v-\nr);
	\pgfmathtruncatemacro{\nr}{\r+4}

	% red path here
	\draw[redpath,medge] (v-\nr) to [bend left] (v-\r);
	\pgfmathtruncatemacro{\r}{\nr}
	\pgfmathtruncatemacro{\nr}{\r+6}
	\draw[redpath,medge] (v-\nr) to [bend left] (v-\r);
	\pgfmathtruncatemacro{\r}{\nr}
	\pgfmathtruncatemacro{\nr}{\r+3}
	\draw[redpath,medge8] (v-\nr) to (v-\r);
	\draw[redpath,medge] (t) to[out=315,in=315] (v-\nr);
	
	% orange
	\pgfmathtruncatemacro{\o}{(\n-6)*2}
	\pgfmathtruncatemacro{\no}{\o+3}
	\draw[orangepath,medge] (v-\o) to  (v-\no);
	\pgfmathtruncatemacro{\no}{\o+4}
	\draw[orangepath,medge] (v-\no) to [bend left] (v-\o);
	\pgfmathtruncatemacro{\o}{\no}
	\pgfmathtruncatemacro{\no}{\o+2}
	\draw[orangepath,medge] (v-\no) to [bend left] (v-\o);
	\pgfmathtruncatemacro{\o}{\no}
	\pgfmathtruncatemacro{\no}{\o-1}
	\draw[orangepath,medge] (v-\no) to (v-\o);
	\pgfmathtruncatemacro{\o}{\no}
	\pgfmathtruncatemacro{\no}{\o+2}
	\draw[orangepath,medge] (v-\no) to [bend left] (v-\o);
	\pgfmathtruncatemacro{\o}{\no}
	\pgfmathtruncatemacro{\no}{\o+2}
	\draw[orangepath,medge] (v-\no) to [bend left] (v-\o);
	\pgfmathtruncatemacro{\o}{\no}
	\pgfmathtruncatemacro{\no}{\o+1}
	\draw[orangepath,medge8] (v-\no) to (v-\o);
	\draw[orangepath,medge] (t) to[out=225,in=225] (v-\no);
	
	% purple
	\pgfmathtruncatemacro{\o}{\no}
	\pgfmathtruncatemacro{\no}{\o+3}
	\draw[purpleedges,medge5] (v-\no) to[ bend right] (v-\o);
	\draw[purpleedges,medge5] (v-\o) to[ bend right] (v-\no);
	\pgfmathtruncatemacro{\r}{\nr}
	\pgfmathtruncatemacro{\nr}{\r+1}
	\draw[purpleedges,medge7] (v-\nr) to[ bend right] (v-\r);
	\draw[purpleedges,medge7] (v-\r) to[ bend right] (v-\nr);
	
	\def\tikzsegment#1#2#3{ % This is the macro explained above
		\path let
		\p1=($(#3)-(#2)$),
		\n1={veclen(\p1)*0.7}
		in (#2) -- (#3) 
		node[minimum width=\n1, 
		inner sep=0pt, 
		pos=0.5,sloped,rectangle,rounded corners,
		#1] 
		(line){};
	}
	
	\newcommand{\drawCut}[2] {
		%\tikzsegment{very thick,dash dot,draw,minimum height=0.31cm}{v-#1}{v-#2};
	}
	
	\foreach \i in {0,1} {
		\pgfmathtruncatemacro{\a}{13+\i * 16};
		\pgfmathtruncatemacro{\b}{16 + \i * 16};
		\drawCut{\a}{\b};
		\pgfmathtruncatemacro{\a}{25+\i * 16};
		\pgfmathtruncatemacro{\b}{20 + \i * 16};
		\drawCut{\a}{\b};
	}

	%hard instance
	\draw[thick] (15.5, 16) ellipse (150pt and 200pt);
	\node [vertex] (c1) at (14, 12) {};
	\node [vertex] (c2) at (17, 12) {};
	\node at (13.5, 12) {$c_1$};
	\node at (16.5, 12) {$c_2$};
	\node at (13.7, 0.9) {$c_3$};
	\node at (17.3, 0.9) {$c_4$};
	\draw[medge7] (c1) to (v-28);
	\draw[medge7] (c2) to (v-24);
	\draw[medge] (v-24) to (c1);
	\draw[medge] (v-28) to (c2);
	\node at (13, 6.5) {Connector gadget $C$};
	\node at (7, 16) {\textsc{Hamiltonian Path} instance $H$};
	\node[vertex] (p1) at (21, 16) {};
	\node[vertex] (p2) at (18, 19) {};
	\node[vertex] (p3) at (17, 15) {};
	\node[vertex] (p4) at (15.5, 20) {};
	\node[vertex] (p5) at (10, 16.5) {};
	\node[vertex] (p6) at (12, 14) {};
	\draw[dedge] (c2) to (p1);
	\draw[dedge] (p1) to (p2);
	\draw[dedge] (p2) to (p3);
	\draw[dedge] (p3) to (c1);
	\draw[dedge] (c1) to (p4);
	\draw[dedge] (p4) to (p5);
	\draw[dedge] (p5) to (p6);

	%connector

\end{tikzpicture}
	
	\caption{\label{fig:diameter_reduction} The illustration of the hardness reduction of Lemma~\ref{lemma:diameter_np_complete}. While, in fact, the graph $G_\ell$ used in the reduction must have $\ell \ge 17$, here we use the graph $G_2$ for clarity.}
\end{figure}

    \begin{claim}\label{claim:two-connectivity}
       $G$ is 2-strongly-connected. 
    \end{claim}
    \begin{proof}[Proof of Claim~\ref{claim:two-connectivity}]
        Clearly, $G$ is strongly connected. Now, assume we remove a vertex $v$ from $G$, $v$ either belongs to an induced copy of $H$ or to an induced copy of $G_\ell$.
        In the first case, any other vertex of $H$ is still reachable from $c_1$ or $c_2$ (whichever is not removed) and vice versa, since $H$ is 2-connected. There is also a path from any vertex of $G_\ell$ to $\{c_1, c_2\} \setminus \{v\}$ and back, since the graph $G[V(G_\ell) \cup \{c_1, c_2\} \setminus \{v\}]$ is unchanged.
        The second case is identical.
    \end{proof}
    \begin{claim}\label{claim:diam}
       $\diam(G) = \diam(G_\ell)$.
    \end{claim}
    \begin{proof}[Proof of Claim~\ref{claim:diam}]
        Clearly, $\diam(G) \ge \diam(G_\ell)$. To show the opposite direction, first, observe that $\diam(H) \le \frac{n'}{2} = 2\ell + \frac{k - 5}{2}$.
        That holds since, for any two vertices $u$ and $v$ in $H$, $\dist_H(u, v)$ is at most $\frac{|H|}{2}$,
        as there are two disjoint paths going from $u$ to $v$, and it cannot be that they both are longer than $\frac{|H|}{2}$.
        For any two vertices $u$ and $v$ inside $G_\ell$, the distance from $u$ to $v$ is unchanged from $G_\ell$, since it is impossible to go
        from $u$ to $H - V(C)$ and then back to $v$ in $G_\ell$, and the vertices of $C$ do not change the distances in $G_\ell$.
        Finally, for a vertex $u \in V(G_\ell)$ and a vertex $v \in V(H)$, observe that a path from $u$ to $v$ necessarily goes through the arc $wt$,
        where either $w = c_3$, $t = c_1$, or $w = c_4$, $t = c_2$.
        The subpath $(u, w)$ is a shortest such path inside $G_\ell$, and by Lemma~\ref{lemma:Gk_paths} its length is at most $4\ell + 15$.
        The subpath $(t, v)$ is a shortest path inside  $H$, and its length is at most $\diam(H) \le 2\ell + \frac{k-5}{2}$.
        Thus, $\dist_G(u, v) \le 6\ell + \frac{k - 5}{2} + 16 \le 8\ell + 10 = \diam(G_\ell)$, since by construction $2\ell \ge k/2 + 34$. The case $u \in V(H)$, $v \in V(G_\ell)$ is symmetrical.
    \end{proof}
    \begin{claim}\label{claim:equivalency}
       The longest path in $G$ has length at least $\diam(G) + k$ if and only if there is a Hamiltonian path in $H$ starting in $w$.
    \end{claim}
    \begin{proof}[Proof of Claim~\ref{claim:equivalency}]
        By Lemma~\ref{lemma:Gk_paths}, no path inside $G_\ell$ has length more than $\diam(G) + 4$.
        Since $|V(H)| = 4\ell + (k - 5)$, there cannot be a path of length more than $4\ell + (k - 6) < 8\ell + 10 + k$ inside $H$.
        Thus, if there is a path $T$ longer than $\diam(G) + k$ in $G$, it must use vertices in both $G_\ell$ and $H$.
        By the structure of $C$ any such path either lies completely inside $H$ or $G_\ell$ while taking only one extra vertex of $C$, or crosses from $G_\ell$ to $H$ through $C$ only once (or, from $H$ to $G_\ell$, but this case is symmetrical).
        In the first case, if the path starts and ends in $H$, its length is at most $4\ell + (k - 5)$.
        Now consider the case where $T$ starts and ends in $G_\ell$.
        If the vertex $h_7$ of the median hat gadget does not lie on $T$,
        $T$ can be transformed into a path $T'$ of the same length that lies completely inside $G_\ell$,
        by replacing the outer vertex of $C$ ($c_1$ or $c_2$) by $h_7$.
        Otherwise, if $T$ contains the vertex $h_7$, it has to start or end in this vertex, as the only two neighbors $h_6$ and $h_8$ of $h_7$ lie also on $T$ separated by the outer vertex of $C$.
        Then by Lemma~\ref{lemma:Gk_paths} the length of $T$ is at most $4\ell + 18 < 8\ell + 10 + k$, as $T$ is a path inside $G_\ell$ that starts or ends in $h_6$ or $h_8$ of the median hat gadget, plus three extra arcs.

        Therefore, the only option when $T$ can have length at least $\diam(G) + k$ is when it has the following structure: it starts at a vertex $u \in V(G_\ell)$, then continues inside $G_\ell$ until it takes the arc $wt$ in $C$, and then takes a final subpath $(t, v)$ inside $H$.
        Here either $w = c_3$, $t = c_1$,
        or $w = c_4$, $t = c_2$, and we drop the completely symmetrical case where the path goes from $H$ to $G_\ell$ through $C$.
        By Lemma~\ref{lemma:Gk_paths}, the length of the $(u, w)$-subpath is at most $4\ell + 15$.
        Now if the length of $T$ is at least $\diam(G) + k = 8\ell + 10 + k$, the subpath $(t, v)$ must be a Hamiltonian path in $H$ since $|V(H)| = 4\ell + (k - 5)$, and the length of $T$ is exactly $8\ell + 10 + k$.

        In the other direction, if there is a Hamiltonian path $P$ in $H$, consider its starting vertex $w$, and the instance of \probLPDiam constructed with this choice of $w$.
        By Lemma~\ref{lemma:Gk_paths}, there is a path of length $4\ell + 15$ inside $G_\ell$
        that ends in $c_4$, which is also the vertex $h_8$ of the median hat gadget. Continuing this path through the arc $c_4c_2$ and then along the Hamiltonian path (recall that $c_2$ in $C$ is $w$ in $H$), we obtain a path of length $8\ell + 10 + k$.
    \end{proof}

    Clearly, the lemma follows from the three claims above.
\end{proof}

\section{Conclusion}\label{sec:concl}
 We proved that if $\mathcal{C}$ is a class of directed graph such that \probDP is in \classP  on $\mathcal{C}$ for $p=3$, then \probLD is \classFPT on $\mathcal{C}$. However  \probDP is \classNP-complete on directed graphs for every fixed $p\geq 2$~\cite{FortuneHW80}.
This leaves open the question of Bez{\'{a}}kov{\'{a}} et al.~\cite{BezakovaCDF19}  about parameterized complexity of \probLD on general directed graphs.  Even the complexity (\classP versus \classNP) of deciding whether a directed graph contains an $(s,t)$-path longer than  $\dist_G(s,t)$  (the case of $k=1$) remains open. Notice that \probLD is not equivalent to \probDP for $p=3$ and, therefore, the hardness of \probDP does not imply hardness of  \probLD. 

Our result implies, in particular, that \probLD is \classFPT on planar directed graphs. There are various classes of directed graphs on which \probDP is tractable for fixed $p$ (see, e.g., the book of Bang-Jensen and Gutin~\cite{BangJensenG09}). 
For example,  by Chudnovsky, Scott, and Seymour~\cite{ChudnovskySS14},  \probDP can be solved in polynomial time for every fixed $p$ on semi-complete directed graphs. Together with Theorem~\ref{thm:main}, it implies that \probLD is \classFPT on semi-complete directed graphs and tournaments.  However, from what we know, these results could be too weak in the following sense. 
Using the structural results of Thomassen~\cite{Thomassen80a}, Bang-Jensen, Manoussakis, and Thomassen in~\cite{Bang-JensenMT92} gave a polynomial-time algorithm to decide whether a semi-complete directed graph has a Hamiltonian $(s,t)$-path for two given vertices $s$ and $t$. Thus the real question is whether  \probLD  is in \classP  on semi-complete directed graphs or tournaments.

The second part of our results is devoted to \probLPDiam.
We proved that this problem is \classNP-complete for general graphs for $k=1$ and showed that it is in \classFPT when the input graph is undirected and 2-connected.
We established the complexity dichotomy for \probLPDiam for the case of 2-strongly-connected directed graphs by showing that the problem can be solved in polynomial time for $k\le 4$ and is \classNP-complete for $k\ge 5$.
This naturally leaves an open question for larger values of strong connectivity. The computational complexity of \probLPDiam on $t$-strongly connected graphs for $t\ge 3$ is open. For a very concrete question, is there a polynomial algorithm for \probLPDiam with $k=5$ on graphs of strong connectivity $3$?

\end{document}